\documentclass[lettersize,journal]{IEEEtran}
\usepackage{amsthm}
\usepackage{amsfonts}
\usepackage{amssymb}
\usepackage{stfloats}
\usepackage{cite}
\usepackage{graphicx}
\usepackage{psfrag}
\usepackage{subfigure}
\usepackage{amsmath}
\usepackage{array}
\usepackage{eurosym}
\usepackage{setspace}
\usepackage{algorithmic}
\usepackage[linesnumbered,ruled,vlined]{algorithm2e}
\usepackage{bm}
\usepackage{multicol}
\usepackage{lipsum} 
\usepackage{multirow}
\usepackage{amsthm}
\usepackage{hyperref}
\hypersetup{hypertex=true,
	colorlinks=true,
	linkcolor=blue,
	anchorcolor=blue,
	citecolor=blue}
\usepackage{enumerate}
\usepackage{nicematrix}
\usepackage{booktabs}
\usepackage{threeparttable}
\usepackage{tikz}
\usepackage{pgfplots}
\usepackage{eurosym}
\usepackage{makecell}
\usepackage{colortbl} 
\usepackage{microtype}
\usepackage{rotating} 
\usepackage{tikz}
\usepackage{xcolor}
\usepackage{boldline} 
\tikzset{every picture/.style={remember picture}}

\DeclareMathOperator*{\argmin}{arg\,min}
\DeclareMathOperator*{\sgn}{sgn}

\DeclareMathOperator*{\avg}{avg}

\DeclareMathOperator*{\HD}{HD}
\DeclareMathOperator*{\OFG}{\pi_{\mathrm{o}}}
\DeclareMathOperator*{\set}{{\tilde{\pi}}}

\ifCLASSINFOpdf
\else
\fi
\definecolor{myskyblue}{RGB}{39,156,191}
\definecolor{myred}{RGB}{245, 57, 61}
\definecolor{mygreen}{RGB}{21, 165, 112}
\definecolor{myblued}{RGB}{0,114,189}
\definecolor{mypurple}{RGB}{126,47,142}
\definecolor{myyellow}{RGB}{243,165,93}
\definecolor{tablepipe}{RGB}{138,222,144}

\definecolor{tableread}{RGB}{240,240,240}
\definecolor{tablewrite}{RGB}{195,195,195}
\definecolor{tablew_r}{RGB}{150,150,150}

\def\subparagraph{} 
\usepackage{titlesec}

\newcommand{\mytextsf}[1]{{\small\textsf{#1}}}

\newcommand{\andreas}[1]{\textcolor{green}{ANDREAS: #1}}
\newcommand{\yuqing}[1]{\textcolor{purple}{YUQING: #1}}

\renewcommand{\yuqing}[1]{}
\renewcommand{\andreas}[1]{}

\newcommand{\pf}[1]{\textcolor{blue}{#1}}
\renewcommand{\pf}[1]{{#1}}

\definecolor{bostonunired}{HTML}{CC0000}    
\definecolor{fireenginered}{HTML}{CF202A}   
\definecolor{vermillionred}{HTML}{E34234}   

\definecolor{titlered}{RGB}{212,0,0}

\definecolor{reddeep}{RGB}{178,24,43}
\definecolor{redlight}{RGB}{252,78,42}
\definecolor{redlight2}{RGB}{255,114,111}

\definecolor{redbright}{RGB}{255,24,43}

\definecolor{perfectgreen}{HTML}{4FBF26}    
\definecolor{maygreen}{HTML}{4E9B47}        
\definecolor{mattelime}{HTML}{75AD50}       

\definecolor{indigorainbow}{HTML}{1D3F6E}   
\definecolor{royalazure}{HTML}{1866E1}      
\definecolor{azure}{HTML}{008AFF}           
\definecolor{blueberry}{HTML}{4F86F7}       
\definecolor{fadednavy}{HTML}{242F78}       
\definecolor{navy}{HTML}{000080}            

\definecolor{goldwebgolden}{HTML}{FFD700}   
\definecolor{radioactive}{HTML}{FAE500}     
\definecolor{gold}{HTML}{FFD700}            

\definecolor{vividgamboge}{HTML}{FF9900} 

\definecolor{shinygray}{HTML}{C7C6C6}

\definecolor{slategrey}{HTML}{708090}   
\definecolor{neutralgray}{HTML}{828382} 

\definecolor{mattecharcoal}{HTML}{3B4248} 

\definecolor{graydark}{HTML}{6B6E70} 

\definecolor{charcoal}{HTML}{36454F} 
\definecolor{charcoalShade1}{HTML}{2B373F}
\definecolor{charcoalShade2}{HTML}{263037}
\definecolor{charcoalShade3}{HTML}{20292F}

\definecolor{sunset1}{HTML}{F3E79B}
\definecolor{sunset2}{HTML}{FAC484}
\definecolor{sunset3}{HTML}{F8A07E}
\definecolor{sunset4}{HTML}{EB7F86}
\definecolor{sunset5}{HTML}{CE6693}
\definecolor{sunset6}{HTML}{A059A0}
\definecolor{sunset7}{HTML}{5C53A5}

\definecolor{bluyi1}{HTML}{F7FEAE}
\definecolor{bluyi2}{HTML}{B7E6A5}
\definecolor{bluyi3}{HTML}{7CCBA2}
\definecolor{bluyi4}{HTML}{46AEA0}
\definecolor{bluyi5}{HTML}{089099}
\definecolor{bluyi6}{HTML}{00718B}
\definecolor{bluyi7}{HTML}{045275}

\definecolor{geyser1}{HTML}{008080}
\definecolor{geyser2}{HTML}{70A494}
\definecolor{geyser3}{HTML}{B4C8A8}
\definecolor{geyser4}{HTML}{F6EDBD}
\definecolor{geyser5}{HTML}{EDBB8A}
\definecolor{geyser6}{HTML}{DE8A5A}
\definecolor{geyser7}{HTML}{CA562C}

\definecolor{temps1}{HTML}{009392}
\definecolor{temps2}{HTML}{39B185}
\definecolor{temps3}{HTML}{9CCb86}
\definecolor{temps4}{HTML}{E9E29C}
\definecolor{temps5}{HTML}{EEB479}
\definecolor{temps6}{HTML}{E88471}
\definecolor{temps7}{HTML}{CF597E}

\definecolor{earth1}{HTML}{A16928}
\definecolor{earth2}{HTML}{BD925A}
\definecolor{earth3}{HTML}{D6BD8D}
\definecolor{earth4}{HTML}{EDEAC2}
\definecolor{earth5}{HTML}{B5C8B8}
\definecolor{earth6}{HTML}{79A7AC}
\definecolor{earth7}{HTML}{2887A1}

\definecolor{bold1}{HTML}{7F3C8D}
\definecolor{bold2}{HTML}{11A579}
\definecolor{bold3}{HTML}{3969AC}
\definecolor{bold4}{HTML}{F2B701}
\definecolor{bold5}{HTML}{E73F74}
\definecolor{bold6}{HTML}{80BA5A}
\definecolor{bold7}{HTML}{E68310}
\definecolor{bold8}{HTML}{008695}
\definecolor{bold9}{HTML}{CF1C90}
\definecolor{bold10}{HTML}{F97B72}
\definecolor{bold11}{HTML}{4B4B8F}
\definecolor{bold12}{HTML}{A5AA99}

\definecolor{pastel1}{HTML}{66C5CC}
\definecolor{pastel2}{HTML}{F6CF71}
\definecolor{pastel3}{HTML}{F89C74}
\definecolor{pastel4}{HTML}{DCB0F2}
\definecolor{pastel5}{HTML}{87C55F}
\definecolor{pastel6}{HTML}{9EB9F3}
\definecolor{pastel7}{HTML}{FE88B1}
\definecolor{pastel8}{HTML}{C9DB74}
\definecolor{pastel9}{HTML}{8BE0A4}
\definecolor{pastel10}{HTML}{B497E7}
\definecolor{pastel11}{HTML}{D3B484}
\definecolor{pastel12}{HTML}{B3B3B3}

\definecolor{prism1}{HTML}{5F4690}
\definecolor{prism2}{HTML}{1D6996}
\definecolor{prism3}{HTML}{38A6A5}
\definecolor{prism4}{HTML}{0F8554}
\definecolor{prism5}{HTML}{73AF48}
\definecolor{prism6}{HTML}{EDAD08}
\definecolor{prism7}{HTML}{E17C05}
\definecolor{prism8}{HTML}{CC503E}
\definecolor{prism9}{HTML}{94346E}
\definecolor{prism10}{HTML}{6F4070}
\definecolor{prism11}{HTML}{994E95}
\definecolor{prism12}{HTML}{666666}

\definecolor{vivid1}{HTML}{E58606}
\definecolor{vivid2}{HTML}{5D69B1}
\definecolor{vivid3}{HTML}{52BCA3}
\definecolor{vivid4}{HTML}{99C945}
\definecolor{vivid5}{HTML}{CC61B0}
\definecolor{vivid6}{HTML}{24796C}
\definecolor{vivid7}{HTML}{DAA51B}
\definecolor{vivid8}{HTML}{2F8AC4}
\definecolor{vivid9}{HTML}{764E9F}
\definecolor{vivid10}{HTML}{ED645A}
\definecolor{vivid11}{HTML}{CC3A8E}
\definecolor{vivid12}{HTML}{A5AA99}

\definecolor{cBrewerPaired1}{HTML}{A6CEE3}
\definecolor{cBrewerPaired2}{HTML}{1F78B4}
\definecolor{cBrewerPaired3}{HTML}{B2DF8A}
\definecolor{cBrewerPaired4}{HTML}{33A02C}
\definecolor{cBrewerPaired5}{HTML}{FB9A99}
\definecolor{cBrewerPaired6}{HTML}{E31A1C}
\definecolor{cBrewerPaired7}{HTML}{FDBF6F}
\definecolor{cBrewerPaired8}{HTML}{FF7F00}
\definecolor{cBrewerPaired9}{HTML}{CAB2D6}
\definecolor{cBrewerPaired10}{HTML}{6A3D9A}
\definecolor{cBrewerPaired11}{HTML}{FFFF99}
\definecolor{cBrewerPaired12}{HTML}{B15928}

\definecolor{cBrewerQualPrint1}{HTML}{E41A1C}
\definecolor{cBrewerQualPrint2}{HTML}{377EB8}
\definecolor{cBrewerQualPrint3}{HTML}{4DAF4A}
\definecolor{cBrewerQualPrint4}{HTML}{984EA3}
\definecolor{cBrewerQualPrint5}{HTML}{FF7F00}
\definecolor{cBrewerQualPrint6}{HTML}{FFFF33}
\definecolor{cBrewerQualPrint7}{HTML}{A65628}
\definecolor{cBrewerQualPrint8}{HTML}{F781BF}
\definecolor{cBrewerQualPrint9}{HTML}{999999}

\definecolor{thsViolet1}{HTML}{E4C7F1} 
\definecolor{thsViolet2}{HTML}{9F82CE} 
\definecolor{thsViolet3}{HTML}{4B4B8F} 

\definecolor{thsPink1}{HTML}{DCB0F2} 
\definecolor{thsPink2}{HTML}{FE88B1} 
\definecolor{thsPink3}{HTML}{CC3A8E} 

\definecolor{thsBlue1}{HTML}{4F86F7} 
\definecolor{thsBlue2}{HTML}{0F52BA} 
\definecolor{thsBlue3}{HTML}{332288} 

\definecolor{thsCyan1}{HTML}{88CCEE} 
\definecolor{thsCyan2}{HTML}{66C5CC} 
\definecolor{thsCyan3}{HTML}{367588} 

\definecolor{thsGreen1}{HTML}{4FBF26} 
\definecolor{thsGreen2}{HTML}{11A579} 
\definecolor{thsGreen3}{HTML}{0F8554} 

\definecolor{thsYellow1}{HTML}{FFFF66} 
\definecolor{thsYellow2}{HTML}{FAE500} 
\definecolor{thsYellow2}{HTML}{FFFF00}

\definecolor{thsOrange1}{HTML}{ECDA9A} 
\definecolor{thsOrange2}{HTML}{F7945D} 
\definecolor{thsOrange3}{HTML}{E58606} 

\definecolor{thsRed1}{HTML}{CC503E} 
\definecolor{thsRed2}{HTML}{CC0000} 
\definecolor{thsRed3}{HTML}{670E10} 

\definecolor{thsGray1}{HTML}{B1B6B7} 
\definecolor{thsGray2}{HTML}{828382} 
\definecolor{thsGray3}{HTML}{36454F} 

\usetikzlibrary{spy}

\begin{document}
	\title{High-Throughput Flexible Belief Propagation\\List Decoder for Polar Codes}
\def\figureautorefname{Fig.}
\newtheorem{Thm}{Theorem}
\newtheorem{Pro}{Proof}
\newtheorem{Lem}{Lemma}
\newtheorem{Cor}{Corollary}
\newtheorem{Def}{Definition}
\newtheorem{Exam}{Example}
\newtheorem{Alg}{Algorithm}
\newtheorem{Prob}{Problem}
\newtheorem{Rem}{Remark}
\newcounter{mytempeqncnt}
\renewcommand\thesubsection{\thesection.\Alph{subsection}}
\newcommand\mybox[2][]{\tikz[overlay,region/.style={draw=black}]\node[region, fill=blue!20, inner sep=2pt, anchor=text, rectangle, rounded corners=0mm,#1] {#2};\phantom{#2}}

\newcommand{\blue}{\textcolor[rgb]{0,0,1}}
\newcommand{\red}{\textcolor[rgb]{1.00,0.00,0.00}}

\author{
	Yuqing Ren,
	Yifei Shen,
	Leyu Zhang,
	Andreas Toftegaard Kristensen,
  	Alexios Balatsoukas-Stimming,~\IEEEmembership{Member,~IEEE,}
  	Andreas Burg,~\IEEEmembership{Senior~Member,~IEEE},
	Chuan Zhang,~\IEEEmembership{Senior~Member,~IEEE}\looseness=-1
	\thanks{Y. Ren, Y. Shen, A. T. Kristensen, and A. Burg are with the Telecommunications Circuits Laboratory (TCL),
		\'{E}cole Polytechnique F\'{e}d\'{e}rale de Lausanne (EPFL), Lausanne 1015, Switzerland (email: \{yuqing.ren, yifei.shen, andreas.kristensen, andreas.burg\}@epfl.ch).
		\emph{Corresponding author: Andreas~Burg}.}
	\thanks{Y. Shen, L, Zhang and C. Zhang are with the LEADS of Southeast University,
		the National Mobile Communications Research Laboratory,
		and the Purple Mountain Laboratories, Nanjing 210096, China (email: chzhang@seu.edu.cn).}
  	\thanks{A. Balatsoukas-Stimming is with the Department of Electrical Engineering, Eindhoven University of Technology, 5600 MB Eindhoven, The Netherlands (email: a.k.balatsoukas.stimming@tue.nl).}
}

\markboth{}{Y. Ren \MakeLowercase{\textit{et al.}}: High-Throughput Flexible Belief Propagation List Decoders for Polar Codes}
\maketitle

\begin{abstract}
	\pf{Owing to its high parallelism, belief propagation (BP) decoding is highly amenable to high-throughput implementations and thus represents a promising solution for meeting the ultra-high peak data rate of future communication systems.}
	However, for polar codes, the error-correcting performance of BP decoding is far inferior to that of \pf{the} widely used CRC-aided successive cancellation list (SCL) decoding algorithm.
	To close the performance gap to SCL, BP list (BPL) decoding expands the exploration of candidate codewords through multiple permuted factor graphs (PFGs).
	From an implementation perspective, designing a unified and flexible hardware architecture \pf{for} BPL decoding that supports \pf{various PFGs and code configurations presents a big challenge}.
	In this paper, we propose the first hardware implementation of a BPL decoder for polar codes and overcome the implementation challenge by applying a hardware-friendly algorithm that generates flexible permutations on-the-fly.
	First, we derive the \pf{graph selection gain} and provide a sequential generation (SG) algorithm to obtain a near-optimal PFG set.
	We further prove that any permutation can be decomposed into a combination of multiple fixed routings, and we design a low-complexity permutation network to satisfy the decoding schedule.
	Our BPL decoder not only has a low decoding latency by executing the decoding and permutation generation in parallel, but also supports an arbitrary list size without any area overhead.
	Experimental results show that, \pf{for length-$\bm{1024}$ polar codes with a code rate of one-half, our BPL decoder with $32$ PFGs has a similar error-correcting performance to SCL with a list size of $4$} and achieves a throughput of $\bm{25.63}$ Gbps and an area efficiency of $\bm{29.46}$ Gbps/mm$\bm{^2}$ at SNR~$\bm{=4.0}$~dB, \pf{which is $1.82\times$ and $4.33\times$ faster than the state-of-the-art BP flip and SCL decoders,~respectively}.

\end{abstract}

\section{Introduction}\label{sec:intro}
\IEEEPARstart{P}{olar} codes, proposed by Ar{\i}kan in~\cite{Arikan2008Channel}, have become an integral part of 5G new radio (NR), where they were ratified as the standard codes for the control channels of 5G enhanced mobile broadband (eMBB) scenarios~\cite{5Gembb}.
Along with the invention of polar codes, Ar{\i}kan introduced successive cancellation (SC) decoding and belief propagation (BP) decoding.
Following the evolution of communication scenarios, both SC and BP decoding led the development of polar decoding algorithms and implementations, which were extended into a series of advanced polar decoders such as SC list (SCL)~\cite{tal2015list,niu2012crc,hashemi2016TCAS1,FastSSCL2017,hanif2018fast,shen2022,ren2022sequence}, BP list (BPL)~\cite{Ahmed2018,elkelesh2018belief,A2019Com,Ren2019ASICON,Ahmed2020CRCBPL,Ren2020,Bai2020ISIT,Nghia2018,feng2021novel}, and BP flip (BPF)~\cite{yu2019belief,Shen2020Tcas,shen2020improved,Ji2020TCAS1} decoders.

While the original SC decoding algorithm can achieve channel capacity at infinite code lengths, it shows poor error-correcting performance with practical finite code lengths.
To improve the error-correcting performance of SC decoding, SCL decoding was proposed in~\cite{tal2015list} to keep a list of up to $\mathbb{L}$ candidate codewords.
Additionally, when concatenated with cyclic redundancy check (CRC) codes~\cite{niu2012crc}, polar codes with SCL decoding outperform low-density parity-check (LDPC) and Turbo codes in terms of the error-correcting performance~\cite{Stimming2017WCNC}.
To satisfy the low-latency and high throughput requirements of eMBB scenarios, node-based fast SCL decoders~\cite{hashemi2016TCAS1,FastSSCL2017,hanif2018fast,shen2022,ren2022sequence} focus on exploiting special constituent codes~\cite{alamdar2011simplified,sarkis2014fast,hanif2017fast,condo2018generalized,zheng2021threshold}, which help to avoid traversing the lower stages of the decoding tree to provide a significant reduction in decoding latency compared to conventional bit-wise SCL decoders.
The state-of-the-art (SOA) node-based SCL decoder~\cite{ren2022sequence} with a list size ($L=8$) achieves a throughput of more than $2.94$~Gbps, which fits the reliability, latency, and throughput requirements of eMBB scenarios.
However, when considering the ultra-high peak data rate requirements of future communication systems~\cite{you2021towards}, SC-based decoders become impractical due to the serial processing inherent in these algorithms~\cite{hashemi2016TCAS1,FastSSCL2017,hanif2018fast,shen2022,ren2022sequence}.

In contrast to SC-based decoding, BP decoding is an inherently parallel algorithm.
BP decoding can \pf{thus} be implemented easily in a multi-stage factor graph in pursuit of a much higher throughput~\cite{Youn14belief}.
Additionally, BP decoding has the potential to realize iterative detection and decoding to achieve better system performance than separate detection and decoding~\cite{Jing2017WCSP,Amin2020TWC}, which further raises the interest in BP decoding for academia and industry.
Though the error-correcting performance of BP decoding improves as the iteration number increases, it is still far behind the SCL performance.
BPF decoding~\cite{yu2019belief,Shen2020Tcas,shen2020improved,Ji2020TCAS1} and BPL decoding~\cite{Ahmed2018,elkelesh2018belief,A2019Com,Ren2019ASICON,Ahmed2020CRCBPL,Ren2020,Bai2020ISIT,Nghia2018,feng2021novel} are two advanced BP algorithms  that can approach the performance of SCL by expanding the exploration of candidate codewords.
BPF decoding guesses the positions of error-prone bits and sequentially corrects them in additional decoding attempts.
Unfortunately, online identification of error-prone bits~\cite{Shen2020Tcas,shen2020improved,Ji2020TCAS1} through sorting and post-processing of channel messages increases the hardware complexity and degrades the maximum operating frequency~\cite{Ji2020TCAS1}.
Alternatively, BPL decoding proposed in~\cite{elkelesh2018belief} tries to decode on multiple permuted factor graphs (PFGs), where the number of possible PFGs is $n!$ $(n=\log_2{N})$ for length-$N$ polar codes.
Decoding schedules of BPL can be divided into parallel~\cite{elkelesh2018belief,A2019Com,Ren2019ASICON,Ahmed2020CRCBPL} and serial schedules~\cite{Nghia2018}, respectively.
In parallel BPL decoding, $\mathbb{L}$ independent BP decoders operate in parallel (each BP decoder works on a unique PFG) and the optimal codeword with the minimum Euclidean distance to the received signals is selected from the $\mathbb{L}$ identified candidate codewords.
However, parallel BPL decoding has very poor hardware utilization, especially for large list sizes.
To avoid the high hardware consumption caused by the parallel architecture, the authors of~\cite{Nghia2018} proposed a serial BPL decoding schedule, in which shuffling the input LLRs can be substituted for permutations of the factor graph stages.
This hardware-friendly decoding strategy allows BPL decoding to reuse a single BP decoder at the cost of merely shuffling the input LLRs into a specific order for each PFG.

To improve the error-correcting performance of BPL decoding, numerous researchers have explored methods of optimizing the PFG selection, including empirical methods~\cite{Hussami2009,elkelesh2018belief,Nghia2018} and analytical methods~\cite{Ren2020,Ahmed2020CRCBPL,Bai2020ISIT}.
It is noteworthy that the authors of~\cite{Ahmed2020CRCBPL} first derived the permutation gain for parallel BPL decoding, which provides the inspiration for analytically solving the optimal PFG selection.
In view of hardware implementations, many works of BP decoders have been presented in \cite{Pamuk2012An,Youn14belief,Bo2014Early,Yang2016ISCAS,Sun2016ISCAS,Abbas17High,chen2019}.
Compared to the classical single-column BP architectures~\cite{Youn14belief}, the SOA double-column bidirectional-propagation architecture~\cite{chen2019} instantiates two processing element (PE) arrays and propagates the left-to-right and right-to-left messages simultaneously to improve the throughput.
Moreover, the most challenging task for the BPL decoder is the implementation of flexible permutations since the PFG selection algorithms~\cite{Ren2020,Ahmed2020CRCBPL,Bai2020ISIT,geiselhart2022polar} are generally dynamic, corresponding to varying code configurations or channel environments.
Even if based on area-efficient serial decoding, the BPL decoder still needs to support the generation of flexible permutations by shuffling the input LLRs into a specific order for each PFG.
A straightforward method is to utilize the Bene\v{s} network~\cite{benevs1964optimal,Daesum2010TVLSI}, which is an optimal non-blocking network that can achieve any arbitrary permutation.
However, the design space of permutations is $n!$ instead of $N!$ for length-$N$ polar codes, and the control signals of the Bene\v{s} network are difficult to generate on-the-fly for each PFG.
It is not efficient to adopt the Bene\v{s} network in the BPL decoder.
In summary, there are thus two critical problems for the BPL~decoder:
\begin{itemize}
	\setlength{\itemsep}{0pt}
	\setlength{\parsep}{0pt}
	\setlength{\parskip}{0pt}
	\item How to select the optimal PFG set from $n!$ PFGs for length-$N$ polar codes?
	\item How to efficiently implement flexible permutations for BPL decoding in hardware?
\end{itemize}
It is further noteworthy that BPL decoding is a particular case of a generalized automorphism ensemble (AE) decoding, in which we can deploy the SC, SCL, or BP decoding on multiple PFGs to achieve ML performance of polar or Reed-Muller (RM) codes~\cite{Geise2021TCOM,bioglio2022group}.
Hence, the solutions to these two problems are significant for both BPL and for generalized AE decoding.

\subsection*{Contributions:}
In this paper, we present the first BPL implementation, which solves the aforementioned two problems, i.e., the use of near-optimal PFG sets and the generation of flexible permutations.
Our contributions comprise the following:
\begin{itemize}
	\setlength{\itemsep}{0pt}
	\setlength{\parsep}{0pt}
	\setlength{\parskip}{0pt}
	\item
	We derive the block error probability of serial BPL decoding and present \pf{a criterion for determining} the best PFG set.
	Then, we propose a sequential generation (SG) algorithm \pf{that can efficiently} obtain a near-optimal PFG set.
	Simulations show that our BPL decoder with $\mathbb{L}=32$ achieves similar error-correcting performance to SCL with~$\mathbb{L}=4$.

	\item
	We propose a hardware-friendly algorithm using low-complexity matrix decomposition to generate flexible permutation routings for all PFGs.
	To this end, we provide a mathematical model for permutations and \pf{demonstrate} that the permutation routing of each PFG can be decomposed into a combination of $n-1$ fixed sub-routings.
	This decomposition process can be done~online.

	\item
	We present the first hardware architecture of a BPL decoder, \pf{based on the double-column bidirectional-propagation scheme~\cite{chen2019}, that incorporates the aforementioned flexible permutation generator.}
	To improve the throughput of the BPL decoder, we adopt a decoupled strategy that enables BP decoding and permutation generation to be executed simultaneously.
	It is noteworthy that our decoder can increase the list size arbitrarily without any additional area overhead.
	Synthesis results show that, for $\mathbb{L}=32$, our decoder can achieve a throughput of $25.63$ Gbps with an area efficiency of $29.46$ Gbps/mm$^2$ at SNR~$=4$~dB, which outperforms the SOA BP and BPF decoders~\cite{Abbas17High, chen2019, Su2022JSSCL, Shen2020Tcas, shen2020improved}.
\end{itemize}

The remainder of this paper is organized as follows.
Section~\ref{sec:preli} reviews the background of polar codes, BP decoding, and BPL decoding.
Section~\ref{sec:FER_perform} analyses the permutation gain for serial BPL decoding and presents a graph selection algorithm for a near-optimal PFG set.
In Section~\ref{sec:Theorem_drivation}, a hardware-friendly algorithm for any permutation generation is proposed.
Section~\ref{sec:bpl_flexible_fg_gen} presents our BPL decoder architecture with several advanced techniques.
Section~\ref{sec:implementation} provides our implementation results and compares them with the SOA polar decoders.
Finally, Section~\ref{sec:conclusion} concludes this paper.

\section{Preliminaries}\label{sec:preli}
\emph{Notation}: Throughout this paper, we use the following symbol definitions.
Boldface lowercase letters $\bm{u}$ denote vectors, where $u_{i}$ means the $i$-th element of $\bm{u}$ and $\bm{u}_{i}^{j}$ denotes the sub-vector $[u_{i}\;u_{i+1}\;\hdots\;u_{j}],i\leq j$.
If $i>j$, $\bm{u}_{i}^{j}=\varnothing$.
Boldface uppercase letters $\mathbf{B}$ denote matrices, where $B_{ij}$ and $\bm{b}_{j}$ denote the element at the $i$-th row and $j$-th column of $\mathbf{B}$ and the $j$-th column of $\mathbf{B}$, respectively.
In terms of the factor graph for polar codes with length $N=2^n$, we use $\OFG$, $[m_{0}\;m_{1}\;m_{2}\;\hdots\;m_{n-1}]$, and $[0\;1\;2\;\hdots\;n-1]$ to represent the original factor graph (OFG), its stages, and its stage order, respectively.
Similarly, we use $\pi$, $[m_{\pi^{0}}\;m_{\pi^{1}}\;m_{\pi^{2}}\;\hdots\;m_{\pi^{n-1}}]$, and $[\pi^{0}\;\pi^{1}\;\pi^{2}\;\hdots\;\pi^{n-1}]$ to denote any other PFG, its stages, and its stage order, respectively.
If $\mathcal{L}$ is \pf{a set of $\mathbb{L}$ PFG candidates}, $\{\set_{0}\;\set_{1}\;\hdots\;\set_{\mathbb{L}-1}\}$, $|\mathcal{L}|=\mathbb{L}$ means its cardinality.
Note that all indices related to decoding start from $0$.
\pf{The hard decision function is defined as $\HD(x)=1$ if $x<0$ and $\HD(x)=0$ if $x\geq0$.}
We adopt the following parameters for polar codes, $N$ is the code length, $K$ is the number of \pf{message bits}, $R=K/N$ the code rate, $P$ the number of CRC bits, $K'=K+P$ the number of \pf{information bits} with the CRC bits attached.
The frozen and unfrozen bit set indices are denoted as $\mathcal{F}$ and $\mathcal{A}$, respectively, and we refer to a code as an $(N,K)$ polar code.
\pf{As in this work we only consider polar codes that are concatenated with CRC codes, we use the term \emph{SCL decoding} to refer to CRC-aided SCL decoding for brevity.}

\subsection{Construction and Encoding of Polar Codes}

Given an input bit sequence $\bm{u}$, the encoded vector $\bm{x}$ is generated by ${\bm{x}}={\bm{u}\cdot{\mathit{\mathbf{G}_{N}}}}$, where ${\mathit{\mathbf{G}_{N}}}=
{\mathbf{F}^{\otimes n}}$
denotes the Kronecker power of the kernel~$\mathbf{F}=\left[\begin{smallmatrix} 1 & 0 \\ 1 & 1\end{smallmatrix}\right]$.
Based on the principle of channel polarization~\cite{Arikan2008Channel}, the $N$ bits in $\bm{u}$ correspond to $N$ coordinated bit channels with different reliabilities, where the $K'$ most reliable bit channels transmit unfrozen bits with CRC attached and the remaining $N-K'$ bit channels transmit frozen bits, typically set to a value of 0.
Note that the metric used to determine the bit channel reliability has an impact on $\mathcal{A}$ and influences the performance of polar codes.
For 5G NR~\cite{NR5G}, a universal reliability sequence is applied to formulate $\mathcal{A}$ with $K'$ bits for uplink (UL) and downlink (DL) channels.
Besides, a novel polar code construction framework tailored to a given decoding algorithm based on a genetic algorithm (GenAlg) was introduced in~\cite{Ahmed2019Tcom}, where populations of unfrozen sets evolve based on the error-correcting performance of~a~given~decoder.

\begin{figure}[t]
	\centering
	\includegraphics[width = 0.85\linewidth]{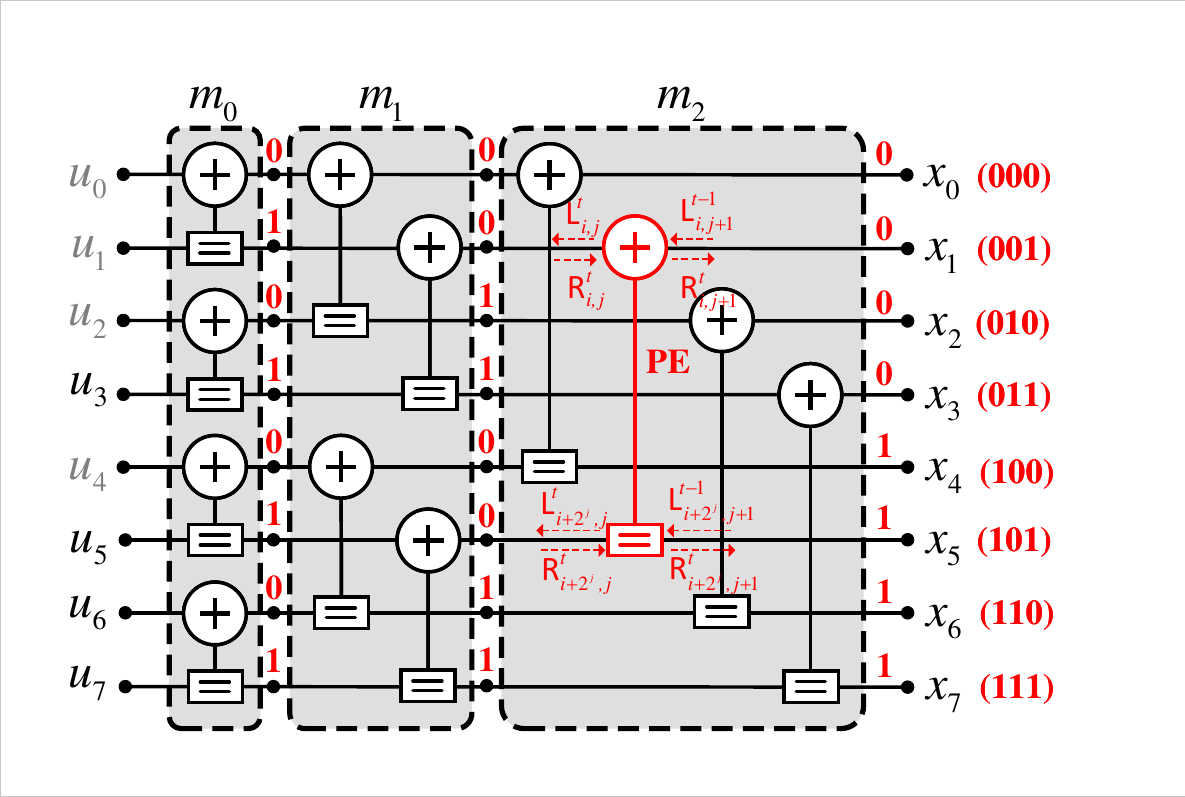}
	\caption{The OFG for length-$8$ polar codes, where one PE is marked in red and $\mathcal{F}=\{0,1,2,4\}$ is marked in grey.}
	\label{fig:contr_8}
\end{figure}

\subsection{BP Decoding of Polar Codes on the Factor Graph}
The BP algorithm is a classical iterative algorithm to calculate the marginal probability by the sum-product (SP) equations on a factor graph~\cite{Kschichang2001FG}.
Motivated by BP decoding for RM codes, Ar{\i}kan first proposed BP decoding for polar codes on the generator matrix-based factor graph~\cite{arikan2010ISBC}.
The OFG structure with three stages [$m_0\;m_1\;m_2$] is shown in Fig.~\ref{fig:contr_8}.
Namely, an $(N,K)$ polar code is represented as an $n$-stage factor graph, and each stage has $N/2$ PEs.
Two types of LLR messages (left-to-right $\mathsf{R}$ and right-to-left $\mathsf{L}$) are propagated over PEs on the factor graph.
\pf{At the $t$-th iteration,} for $j=0,\hdots,n$, \pf{$\mathsf{R}_j^t$} and \pf{$\mathsf{L}_j^t$} can be denoted as the $j$-th column of $\mathsf{R}$- and $\mathsf{L}$-messages, respectively, where \pf{$\mathsf{R}_{i,j}^t$} and \pf{$\mathsf{L}_{i,j}^t$} denote the messages at the $i$-th bit index of the $j$-th column, respectively.
Each PE propagates $\mathsf{R}$- and $\mathsf{L}$-messages as follows~\cite{xu2015xj}.
\begin{equation}\label{eq:bp-itera}
	\pf{\left\{
	\begin{aligned}
		&\mathsf{L}_{i,j}^{t}=g({\mathsf{L}_{i,j+1}^{t-1}},{{\mathsf{L}_{i+2^j,j+1}^{t-1}}+{\mathsf{R}_{i+2^j,j}^{t}}},\beta_{\mathsf{L}}),\\
		&\mathsf{L}_{i+2^j,j}^{t}=g({\mathsf{L}_{i,j+1}^{t-1}},{\mathsf{R}_{i,j}^{t}},\beta_{\mathsf{L}})}+{\mathsf{L}_{i+2^j,j+1}^{t-1},\\
		&\mathsf{R}_{i,j+1}^{t}=g({\mathsf{R}_{i,j}^{t}},{\mathsf{L}_{i+2^j,j+1}^{t-1}+\mathsf{R}_{i+2^j,j}^{t}},\beta_{\mathsf{R}}),\\
		&\mathsf{R}_{i+2^j,j+1}^{t}=g({\mathsf{R}_{i,j}^{t}},{\mathsf{L}_{i,j+1}^{t-1}},\beta_{\mathsf{R}})+\mathsf{R}_{i+2^j,j}^{t}.\\
	\end{aligned}
	\right.}
\end{equation}
where we adopt the offset-MS (OMS) equation~\cite{Shen2020Tcas,dl2020Xu} to approximate the SP equation for all iterative BP decoders
It can be implemented easily in hardware to approach the performance of the SP, where $g(a,b,\beta)=\sgn(a)\cdot\sgn(b)\cdot\max(\min(|a|,|b|)-\beta,0)$ and $[\beta_{\mathsf{R}}\;\beta_{\mathsf{L}}]=[0.25\;0]$.
At the beginning of BP decoding, \pf{$\mathsf{R}_0^{0}$} is initialized as \emph{a-priori} $+\infty$ or $0$ according to the bit channel allocation of $\mathcal{F}$ and \pf{$\mathsf{L}_{n}^{0}$} is initialized as \emph{a-posterior} LLR values from the received signals $\bm{y}$, i.e., $\ln\frac{\mathrm{Pr}(y_i|x_i=0)}{\mathrm{Pr}(y_i|x_i=1)}, 0\leq i\leq N-1$.
$\mathsf{R}$- and $\mathsf{L}$-messages of other stages on the factor graph are initialized as $0$.
When the maximum number of iteration $I_{\max}$ is reached, the HD results $\bm{\hat{u}}$ are estimated based on the decision LLRs \pf{($\mathsf{R}_0^{I_{\max}-1}+\mathsf{L}_0^{I_{\max}-1}$)}.
\pf{In the following, we omit the iteration index $t=0$ of the initial LLRs $\mathsf{R}_0^{0}$ and $\mathsf{L}_{n}^{0}$ for brevity.}

\begin{figure}[t]
	\centering
	\includegraphics[width = \linewidth]{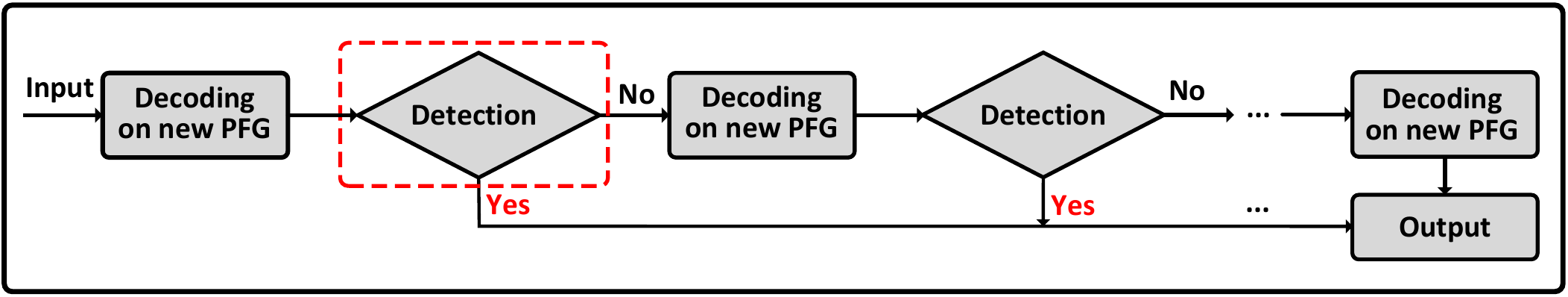}
	\caption{Overall framework for serial BPL decoding with the detection.}
	\label{fig:BPL_framework}
\end{figure}

\subsection{BPL Decoding of Polar Codes}\label{sec:bpl_intro}
BPL decoding~\cite{elkelesh2018belief} is an efficient algorithm to enhance the error-correcting performance of BP decoding, which executes multiple BP decoding procedures on multiple PFGs either in parallel~\cite{elkelesh2018belief,A2019Com,Ren2019ASICON,Ahmed2020CRCBPL} or serially~\cite{Nghia2018}.
Parallel BPL decoding instantiates a set $\mathcal{L}$ of $\mathbb{L}$ independent BP decoders (each BP decoder works on a unique PFG), which leads to a poor hardware utilization since only one result is finally retained.
Alternatively, serial BPL decoding can reuse a single BP decoder, which is illustrated in Fig.~\ref{fig:BPL_framework}.
If a BP decoding attempt \pf{fails to pass} the detection within $I_{\max}$ \pf{iterations}, serial BPL decoding activates the decoding on next PFG.
Note that due to the detection in Fig.~\ref{fig:BPL_framework}, \pf{the miss and error-detection events for} each PFG are introduced, \pf{which are denoted as $M$ and $D$, respectively}.\footnote{\pf{$M$ represents a wrongly estimated information sequence which passes the detection, and $D$ represents the event that fails to pass the detection.
We generally use CRC detection as the detection strategy in serial BPL decoding.}}
In the following, all BPL decoders are used in the serial structure unless stated otherwise.
\pf{With regard to the PFG selection,} previous works have found that the OFG always yields the best error-correcting performance~\cite{Ahmed2020CRCBPL} and the PFGs which fix more left stages and only permute the right-most side of the graph can result in better error-correcting performance~\cite{Nghia2018}.
This can be done as follows
\begin{equation}\label{eq:Sec3_k}
	[m_{0}\;m_{1}\;\hdots\;m_{p-1}\;m_{\pi^{p}}\;\hdots\;m_{\pi^{n-1}}],\;0\leq p\leq n,
\end{equation}
where $p$ denotes the number of fixed left stages.
\pf{It is noticeable that} the size of the design space of PFGs is \pf{reduced from $n!$ to} $(n-p)!$, which facilitates our derivations in Section~\ref{sec:FER_perform}.

\begin{figure*}[t]
	\centering
	\includegraphics[width = 0.9\linewidth]{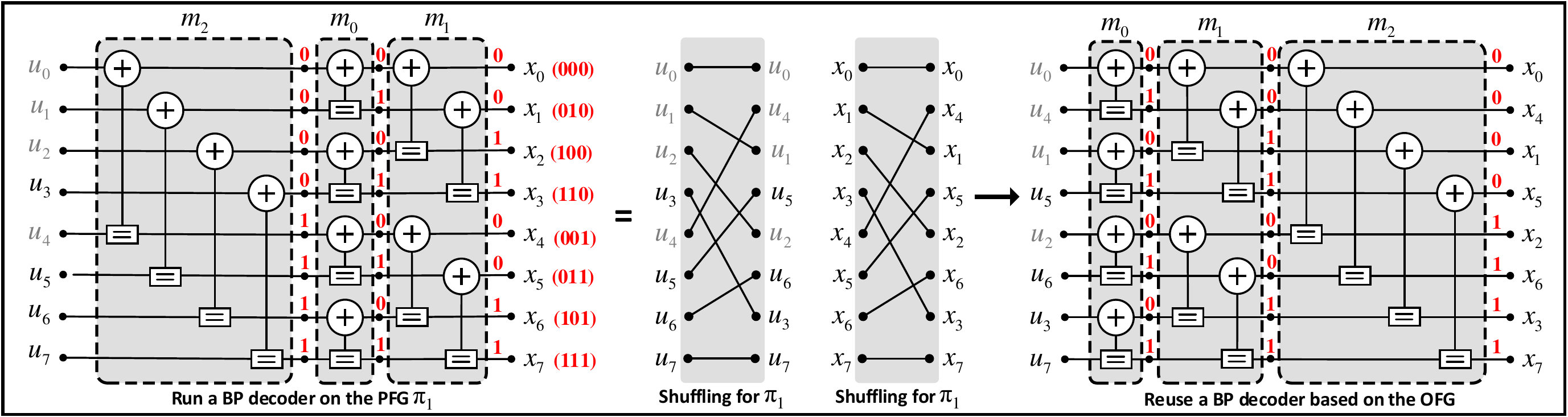}
	\caption{Permutation of a PFG $\pi_1=[m_{2}\;m_{0}\;m_{1}]$ based on the input shuffling.}
	\label{fig:permutations}
\end{figure*}

\subsection{Permutation by Shuffling the Input LLRs}\label{sec:Sec2_D}
\pf{In~\cite{Nghia2018}, it is demonstrated that} a one-to-one mapping exists between permutation of the factor graph stages and shuffling of \pf{the bit indices} in the codeword.
If we consider a node $\bigoplus$ as a `0' and a node $\boxed{=}$ as a `1' on the factor graph, we can derive the `0/1' sequences that describe the factor graph from the binary expansion of the bit-index $i$ for $u_{i}$ and $x_{i}$, as shown in the OFG of Fig.~\ref{fig:contr_8}.
Subsequently, Fig.~\ref{fig:permutations} depicts the permutation of a PFG $\pi_{1}=[m_{2}\;m_{0}\;m_{1}]$ based on the input shuffling.
Note that by permuting the factor graph stages, the positions of the nodes $\bigoplus$ and $\boxed{=}$ have been changed, i.e., the binary expansions of the bit indices for the input LLRs also have been changed.
We can instantiate a routing based on these binary expansions to shuffle the input LLRs, which replaces permutations of the factor graph stages.\footnote{\pf{Note that in a BP decoder, we need to shuffle the input LLRs $\mathsf{R}_0$ and $\mathsf{L}_n$. However, for brevity, we demonstrate how the shuffling works by shuffling the corresponding $\bm{u}$ and $\bm{x}$ in Fig.~\ref{fig:permutations}.}}
After the above input shuffling, we use a single BP decoder based on the OFG to decode on different PFGs.
Due to the varying optimal $\mathcal{L}$ for different code configurations, it is significant for a BPL decoder to support generating the flexible input shuffling, which is discussed in Section~\ref{sec:Theorem_drivation}.

\section{Proposed Graph Selection Algorithm Based on the \pf{Graph Selection Gain}}\label{sec:FER_perform}

In this section, we derive the block error probability of serial BPL decoding and further propose the SG algorithm to obtain a near-optimal PFG set.
Numerical results show that BPL decoding with the SG algorithm and $\mathbb{L}=8$ can yield a similar performance as SCL with $\mathbb{L}=2$ under different code rates of 5G NR polar codes.
Subsequently, to fully show the potential of BP decoding, we employ the construction using the GenAlg~\cite{Ahmed2019Tcom} to allocate bit channels for BP decoding and apply the proposed near-optimal set.
Simulations for this case illustrate that BPL decoding with $\mathbb{L}=32$ achieves the error-correcting performance of SCL decoding with $\mathbb{L}=4$.

\subsection{Permutation Gain Analysis for BPL Decoding}\label{sec:sec3_A}

Due to \pf{the introduced detection module} in Fig.~\ref{fig:BPL_framework}, the individual block error probability $\mathrm{Pr}(E_l)$ for the $l$-th PFG $\set_l$ in $\mathcal{L}$ combines \pf{the miss rate $\mathrm{Pr}(M_l)$ and the error-detection rate $\mathrm{Pr}(D_l)$, which can be expressed as follows}
\begin{equation}\label{eq:Block_err_pfg}
	\pf{\mathrm{Pr}(E_l)=\mathrm{Pr}(M_l)+\mathrm{Pr}(D_l).}
\end{equation}

Considering the serial decoding schedule, we modify the derivation for the block error probability of BPL decoding ${\mathrm{Pr}}({E}_{{\mathrm{BPL}}(\mathcal{L})})$ in~\cite{Ahmed2020CRCBPL}, as illustrated in \eqref{eq:Pr_BPL}
\begin{equation}\label{eq:Pr_BPL}
	\pf{
	\begin{aligned}
		&{\mathrm{Pr}}({E}_{{\mathrm{BPL}}(\mathcal{L})})=\\
		&\underbrace{\sum_{l=1}^{\mathbb{L}-1}{\mathrm{Pr}}\left({M_{l}},\bigcap_{k=0}^{l-1}{D}_{k}\right)+{\mathrm{Pr}}({M}_{0})}_{\mathrm{miss \;probability}}+\underbrace{{\mathrm{Pr}}\left(\bigcap_{l=0}^{\mathbb{L}-1}{D}_{l}\right)}_{\mathrm{list\;error\;probability}},\\
	\end{aligned}
	}
\end{equation}
where the list error probability ${\mathrm{Pr}}\left(\bigcap_{l=0}^{\mathbb{L}-1}{D}_{l}\right)$ reveals the improvement in the error-correcting performance from list decoding.
From~\eqref{eq:Pr_BPL} we also see that the probability of this term ${\mathrm{Pr}}\left(\bigcap_{l=0}^{\mathbb{L}-1}{D}_{l}\right)$ decreases as the number of PFG candidates $\mathbb{L}$ increases.
If $\mathbb{L}=1$, \pf{${\mathrm{Pr}}({E}_{{\mathrm{BPL}}(\mathcal{L})})={\mathrm{Pr}}({M}_{0})+{\mathrm{Pr}}({D}_{0})$.}
In serial BPL decoding, the detection module performs the CRC detection~\cite{Ren2015Asicon} rather than the minimum Euclidean distance used in parallel BPL decoding~\cite{elkelesh2018belief}.
Note that the CRC detection makes the output codeword not necessarily maximum-likelihood decodable and it introduces the miss~probability as~discussed~in~\eqref{eq:Pr_BPL}.

Most PFG selection algorithms in~\cite{Nghia2018, Ren2020, Bai2020ISIT, geiselhart2022polar} are only concerned with \pf{developing} an efficient metric to \pf{identify} PFGs with the minimum ${\mathrm{Pr}}({D}_{l}),\set_l\in\mathcal{L}$.
However, as indicated in~\eqref{eq:Pr_BPL}, the list error probability is determined by the joint block error probability of the selected PFGs, \pf{rather than the individual ${\mathrm{Pr}}({D}_{l})$ of each chosen PFG}.
Therefore, the task of finding the optimal $\mathcal{L}^{\star}$ that yields the lowest block error probability ${\mathrm{Pr}}({E}_{{\mathrm{BPL}}(\mathcal{L})})$ has transformed into an optimization problem to minimize~\eqref{eq:Pr_BPL}, as shown in~\eqref{eq:Pro_argminS}
\begin{equation}\label{eq:Pro_argminS}
	\mathcal{L}^{\star}=\underset{\mathcal{L}}{\argmin\;}{\mathrm{Pr}}({E}_{{\mathrm{BPL}}(\mathcal{L})}),\;\mathrm{s.t.}\;|\mathcal{L}|=\mathbb{L}.
\end{equation}

\begin{figure*}[t]
	\centering
	\input{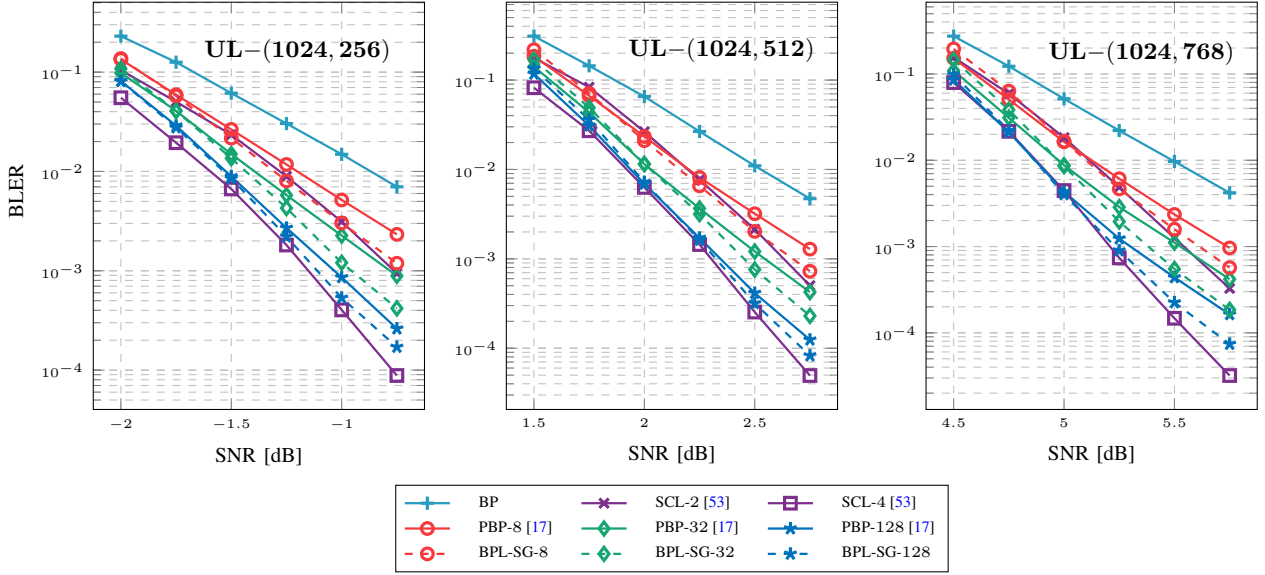}
	\caption{BLER performance of BP, PBP, the proposed BPL-SG, and SCL decoding for 5G UL polar codes with $R=\{\frac{1}{4},\frac{1}{2},\frac{3}{4}\}$, respectively. All iterative decoders use CRC for detection and $I_{\max}=50$.}
	\label{fig:FER_1}
\end{figure*}

It is noteworthy that \pf{${\mathrm{Pr}}({M}_{l})$ of~\eqref{eq:Pr_BPL} is a constant decided only by the CRC polynomial} and is independent of the PFG index $l$.
\pf{Therefore, we refer to them as ${\mathrm{Pr}}_{M}$ in the following discussion.}
To \pf{streamline} the optimization problem in~\eqref{eq:Pro_argminS}, as mentioned in Section~\ref{sec:bpl_intro}, we empirically default $\set_0=\OFG$ since the OFG always yields the lowest block error probability.
We then reformulate~\eqref{eq:Pr_BPL} based on probabilities that are conditioned on ${\mathrm{Pr}}({D}_{0})$ to reflect the PFG selection gain in~\eqref{eq:Pro_BPL_condition}
\begin{equation}\label{eq:Pro_BPL_condition}
	\pf{
	\begin{aligned}
		&{\mathrm{Pr}}({E}_{{\mathrm{BPL}}(\mathcal{L})})={\mathrm{Pr}}_{M}+{\mathrm{Pr}}({D}_{0})\cdot\\
		&\left({\mathrm{Pr}}_{M}\cdot\sum_{l=1}^{\mathbb{L}-1}{\mathrm{Pr}}\left(\left.\bigcap_{k=1}^{l-1}{D}_{k}\right\vert{D}_{0}\right)+\underbrace{{\mathrm{Pr}}\left(\bigcap_{l=1}^{\mathbb{L}-1}{D}_{l}\Bigg|{D}_{0}\right)}_{\mathrm{PFG\;selection\;gain}}\right).\\
	\end{aligned}
	}
\end{equation}

Subsequently, we use $\mathcal{L'}=\mathcal{L}\backslash\OFG$ to denote the remaining PFG set \pf{and further simplify~\eqref{eq:Pr_arg} into~\eqref{eq:Pro_argminS} based on~\eqref{eq:Pro_BPL_condition}}

\begin{equation}\label{eq:Pr_arg}
	\begin{aligned}
		&\mathcal{L}^{\star}
		\approx{\underset{\mathcal{L'}}{\argmin\;}}{\mathrm{Pr}}({E}_{{\mathrm{BPL}}(\mathcal{L})}|D_0)\approx\\
		&\underset{\mathcal{L'}}{\argmin\;}\left(\underbrace{{\mathrm{Pr}}_{M}\cdot\sum_{\pf{l=1}}^{\mathbb{L}-1}{\mathrm{Pr}}\left(\bigcap_{k=1}^{l-1}{D}_{k}\Bigg|{D}_{0}\right)}_{\mathrm{miss\;probability:\;Pr_{\emph{M}(\mathcal{L'})}}}
		+\underbrace{{\mathrm{Pr}}\left(\bigcap_{l=1}^{\mathbb{L}-1}{D}_{l}\Bigg|{D}_{0}\right)}_{\mathrm{PFG\;gain:\;Pr_{PFG(\mathcal{L'})}}}\right),\\
	\end{aligned}
\end{equation}
where $\mathrm{Pr_{\emph{M}(\mathcal{L'})}}$ is a constant $\mathrm{Pr}_{M}$ \pf{if} $\mathbb{L}=2$.
In order to obtain the optimal $\mathcal{L'}$ that minimizes~\eqref{eq:Pr_arg}, we factorize $\mathrm{Pr}_{M(\mathcal{L'})}$ and $\mathrm{Pr_{PFG(\mathcal{L'})}}$ as shown in~\eqref{eq:Pr_M}, where $\mathrm{Pr_{\emph{M}(\mathcal{L'})}}$ and $\mathrm{Pr_{PFG(\mathcal{L'})}}$ both have a similar sequential structure.
\pf{In the case of $\mathbb{L}=2$, ~\eqref{eq:Pr_arg} is simplified to $\mathrm{Pr}(D_{1}|D_{0})+\mathrm{Pr}_{M}$}, and we find $\set_1$ by minimizing $\mathrm{Pr}(D_{1}|D_{0})$.
\pf{As} $\mathbb{L}$ increases from $2$ to $3$ and we reserve the selected $\{\set_0\;\set_1\}$, we can consecutively obtain $\set_2$ by minimizing $\mathrm{Pr}(D_{2}|D_{0},D_{1})$.
Note that when $\mathbb{L}=3$, $\{\set_0\;\set_1\;\set_2\}$ selected by the above method is only an approximation of~\eqref{eq:Pr_arg}, which implies that we sequentially obtain $\set_1$ and $\set_2$ following a greedy algorithm.

\begin{equation}\label{eq:Pr_M}
	\left\{
	\begin{aligned}
		&\begin{aligned}
			&{\mathrm{Pr}_{M(\mathcal{L'})}}={\mathrm{Pr}_{M}}\cdot\\
			&\left(\pf{1+}{\mathrm{Pr}}(D_{1}|{D}_{0})\cdot\left(\hdots\left(1+{\mathrm{Pr}}\left(D_{\mathbb{L}-2}\pf{\Bigg|}\bigcap_{l=0}^{\mathbb{L}-3}D_{l}\right)\right)\right)\right),\\
		\end{aligned}\\
		&\begin{aligned}
			{\mathrm{Pr_{PFG(\mathcal{L'})}}}=&{\mathrm{Pr}(D_{1}|D_{0})}\hdots{\mathrm{Pr}}\left(D_{\mathbb{L}-1}\Bigg|\bigcap_{l=0}^{\mathbb{L}-2}D_{l}\right),\\
		\end{aligned}\\
	\end{aligned}
	\right.
\end{equation}

In conclusion, \pf{we employ a} greedy algorithm applied to the sequential structure of~\eqref{eq:Pr_M} when selecting the $l$-th PFG for $\mathcal{L}$.
Specifically, we always choose the $\set_l$ that can minimize $\mathrm{Pr}(D_{l}|D_{0},D_{1},\hdots,D_{l-1})$, based on the previous $l-1$ selected PFGs, $l\in[1,\mathbb{L})$.
This method decomposes the issue of minimizing the joint block error probability into $\mathbb{L}-1$ consecutive minimization problems of conditional probabilities, which provides a near-optimal result solution for~\eqref{eq:Pr_arg}.

\subsection{Sequential Generation (SG) Algorithm for Graph Selection}\label{sec:mc_sec3}
In this section, we propose a graph selection algorithm called SG, \pf{which can} sequentially perform $\argmin\mathrm{Pr}(D_{l}|D_{0},D_{1},\hdots,D_{l-1})$ to obtain $\set_l$ for $\mathcal{L}$, $l\in[1,\mathbb{L})$.
In addition, \pf{by incorporating the conditional probability} ${\mathrm{Pr}}({D}_{0})$ from~\eqref{eq:Pr_BPL} to~\eqref{eq:Pr_arg}, we observe that ${\mathrm{Pr}}({E}_{{\mathrm{BPL}}(\mathcal{L})}|D_0)\gg{\mathrm{Pr}}({E}_{{\mathrm{BPL}}(\mathcal{L})})$.
Thus, we can use a relatively small dataset to determine the remaining near-optimal $\set_l$ \pf{using} Monte-Carlo~simulations.
\begin{algorithm}[t]
	\caption{\texttt{SG Process}}\label{alg:BPL_MC}
	\KwIn{$\mathcal{D}$, $\mathcal{P}$, $\mathbb{L}$}
	\KwOut{$\mathcal{L}$}
	\tcp{$\!\!$default $\OFG$ is the $0$-th element of $\mathcal{L}$}
	\upshape $\mathcal{L}(0)\leftarrow\OFG$\;
	\For{\upshape $i=0$ to $|\mathcal{P}|-1$}
	{
		\tcp{$\!\!$record the failed frames as 1s}
		$\mathbf{T}(i,:)\leftarrow$\mytextsf{BPEvaluate($\mathcal{P}(i),\mathcal{D}$)}\;
	}
	\tcp{$\!\!$generate the $\mathcal{L}$}
	\For{\upshape $l=1$ to $\mathbb{L}-1$}
	{
		\tcp{$\!\!$find the minimum $\mathrm{Pr}(D_{l}|D_{0},\hdots,D_{l-1})$}
		$i^{\star}\leftarrow\mytextsf{selectBestList}(\mathbf{T})$\;
		$\mathcal{L}(l)\leftarrow\mathcal{P}(i^{\star})$\;
		\tcp{$\!\!$delete the corrected errors}
		$\{\mathbf{T},\mathcal{D}\}\leftarrow$\mytextsf{updateDataset}($\mathbf{T}$, $\mathcal{D}$, $i^{\star}$)\;
	}
\end{algorithm}

To realize the SG algorithm, we first generate a dataset $\mathcal{D}$ \pf{containing} $|\mathcal{D}|$ received vectors $\bm{y}$ that fail to pass the CRC detection under BP decoding on $\OFG$.
Let $\mathcal{P}$ denote the search space of PFGs.
As mentioned in Section~\ref{sec:bpl_intro}, the authors of~\cite{Nghia2018} found that the PFGs which fix more left stages and only permute the right-most side of the graph tend to have a better error-correcting performance.
Hence, we \pf{adjust} $p$ in~\eqref{eq:Sec3_k} to reasonably decrease the design space of $\mathcal{P}$ \pf{to} $(n-p)!$ PFGs and simplify computational complexity.
Subsequently, for each PFG candidate, we evaluate its block error rate (BLER) performance in $\mathsf{BPEvaluate()}$ for all received words in $\mathcal{D}$.
The current frame is recorded as `$1$' in $\mathbf{T}$ if \pf{decoding fails}, and $\mathbf{T}$ is a dynamic matrix with the initial dimension of $(|\mathcal{P}|\times|\mathcal{D}|)$.
In Algorithm~\ref{alg:BPL_MC}, the first element of $\mathcal{L}$ is set to~$\OFG$.
\pf{For} the $i$-th PFG in $\mathcal{P}$ and the $j$-th received codeword in $\mathcal{D}$, the success/failure of the CRC detection is marked as `$0/1$' in $\mathbf{T}(i,j)$ depending on if decoding succeeded or not.
In order to further populate $\mathcal{L}$, we use the function $\mytextsf{selectBestList()}$ to return the index~$i^{\star}$ that corresponds to the minimum weight row in $\mathbf{T}$, which is equivalent to minimizing $\mathrm{Pr}(D_{l}|D_{0},\hdots,D_{l-1})$.
After storing $\mathcal{P}(i^{\star})$ into $\mathcal{L}(l)$, the function $\mytextsf{updateDataset()}$ dynamically updates $\mathbf{T}$ and $\mathcal{D}$ by deleting the columns corresponding to the $0$s in $\mathbf{T}(i^{\star},:)$ and corresponding samples in $\mathcal{D}$ to only keep the erroneous cases.
The above operations (lines $4- 7$ in Algorithm~\ref{alg:BPL_MC}) are repeated for $\mathbb{L}-1$-times until filling $\mathcal{L}$.

\subsection{Numerical Results of the Proposed Algorithm}
Fig.~\ref{fig:FER_1} compares the BLER performance of BP decoding, permuted BP (PBP) decoding~\cite{Nghia2018}, SCL decoding~\cite{Stimming_2015}, and the proposed BPL decoding with the SG algorithm (BPL-SG) for 5G UL polar codes with $N=1024$ and $R\in \{\frac{1}{4}\;\frac{1}{2}\;\frac{3}{4}\}$, where all iterative decoders have the same $I_{\max} = 50$. 
In all following captions, the tailored ``-$\mathbb{L}$'' denotes the employed list size.
\pf{For each case (i.e., SNR$=-1$ dB~for $R=\frac{1}{4}$, SNR$=2.5$ dB for $R=\frac{1}{2}$, and SNR$=5.5$ dB for $R=\frac{3}{4}$),} we generate a relatively small dataset $\mathcal{D}$ with $10^{6}$ samples that fail to pass the CRC detection under BP decoding on $\OFG$.
However, as length-$1024$ polar codes contain $10!$ PFGs, it is impractical to simulate all of them.
In order to make a reasonable trade-off between computational complexity and performance, we set $p=4$ in~\eqref{eq:Sec3_k} and fix the left stages as $[m_{0}\;m_{1}\;m_{2}\;m_{3}]$ to reduce the cardinality of $\mathcal{P}$, which now contains only $720$ PFG candidates.
Numerical results from simulating over AWGN channels with binary phase-shift keying (BPSK) modulation show that BPL-SG-$32$ provides a $0.1$ dB improvement in comparison with PBP-$32$ at BLER~$=10^{-3}$, \pf{and this improvement increases to $0.2$ dB for $R=\frac{1}{4}$.}
When $\mathbb{L}$ increases to $128$, BPL-SG further approaches the performance of SCL-$4$.
\pf{To show the advantage of a serial schedule in terms of computational complexity, we compare the average number of iterations $I_{\mathrm{avg}}$ for two schedules that both use the SG algorithm in Fig.~\ref{fig:sec3_twoschedule}.
Note that there is no termination within each single BP decoder to make a fair comparison.
For $\mathbb{L}=128$, a serial architecture can reduce around $99.2\%$ of the iterations, demonstrating the superiority of the serial schedule for hardware implementation.
}

\begin{figure}[t]
	\centering
	\usepgfplotslibrary{groupplots}

\begin{tikzpicture}[spy using outlines]
\usepgflibrary{decorations.pathmorphing}
\usepgflibrary[decorations.pathmorphing]
\usetikzlibrary{decorations.pathmorphing}
\usetikzlibrary[decorations.pathmorphing]
\definecolor{myblued}{RGB}{0,114,189}
\definecolor{myred}{RGB}{217,83,25}
\definecolor{myredd}{RGB}{248,74,173}
\definecolor{myyellow}{RGB}{237,137,32}
\definecolor{mypurple}{RGB}{126,47,142}
\definecolor{myblues}{RGB}{77,190,238}
\definecolor{mygreen}{RGB}{21, 165, 112}
\pgfplotsset{
    label style = {font=\fontsize{9pt}{7.2}\selectfont},
    tick label style = {font=\fontsize{6pt}{7.0}\selectfont}
  }
\usetikzlibrary{
    matrix,
}
\begin{axis}[
scale = 1,
xmin=1.5,xmax=3.125,
ymin=0,ymax=1000,
scale only axis,
ylabel={$I_{\mathrm{avg}}$}, ylabel style={yshift=-1em},
xlabel={SNR [dB]}, xlabel style={yshift=0.5em},
xtick={1.75, 2.0, 2.25, 2.5, 2.75, 3.0},
xticklabels={$1.75$, $2.0$, $2.25$, $2.5$, $2.75$, $3.0$},
ytick={50, 100, 300, 600, 760, 910},
yticklabels={$50$, $100$,$300$,$600$,$1600$,$6400$},
ymajorgrids=true,
xmajorgrids=true,
grid style=dashed,
separate axis lines,
y axis line style= { draw opacity=0 },
xshift=-1\columnwidth,
width=7.2cm, height=4cm,
legend style={
	nodes={scale=1, transform shape},
    legend columns=3,
    at={(0.5,-0.33)},
    anchor={center},
    cells={anchor=west},
    column sep= 0mm,
    row sep= -0.25mm,
    font=\fontsize{6.9pt}{7.2}\selectfont,
},
legend columns=3,
]
 
\addplot[
color=myblued,
mark=o,
very thick,
mark size=2.5,
line width=0.4mm,
]
table {
1.625	400
1.75	400
1.875	400
2	400
2.125	400
2.25	400
2.375	400
2.5	400
2.625	400
2.75	400
2.875	400
3	400

};
\addlegendentry{Paral. BPL-SG-8}

\addplot[
color=myblued,
mark=square,
very thick,
mark size=2.5,
line width=0.4mm,
]
table {
	1.625	760
	1.75    760
	1.875   760
	2       760
	2.125	760
	2.25	760
	2.375	760
	2.5     760
	2.625	760
	2.75	760
	2.875	760
	3	    760
	
};
\addlegendentry{Paral. BPL-SG-32}

\addplot[
color=myblued,
mark=diamond,
very thick,
mark size=3,
line width=0.4mm,
]
table {
	1.625	910
	1.75    910
	1.875	910
	2	    910
	2.125	910
	2.25	910
	2.375	910
	2.5	    910
	2.625	910
	2.75	910
	2.875	910
	3	    910
	
};
\addlegendentry{Paral. BPL-SG-128}

\addplot[
color=myred,
mark=o,
very thick,
mark size=2.5,
line width=0.4mm,
]
table {
	1.625	104.5
	1.75	85.95
	1.875	69.7625
	2	61.61666667
	2.125	57.23611111
	2.25	54.355
	2.375	52.39
	2.5	51.635
	2.625	50.9275
	2.75	50.665
	2.875	50.425
	3	50.29
};
\addlegendentry{Serial BPL-SG-8}

\addplot[
color=myyellow,
mark=square,
line width=0.4mm,
very thick,
mark size=2.5,
]
table {
	1.625	230.45
	1.75	160.85
	1.875	106
	2	77.62777778
	2.125	65.85
	2.25	58.755
	2.375	54.5525
	2.5	53.195
	2.625	51.7925
	2.75	51.3925
	2.875	50.825
	3	50.46
};
\addlegendentry{Serial BPL-SG-32}

\addplot[
color=mygreen,
mark=diamond,
very thick,
mark size=3,
line width=0.4mm,
]
table {
	1.625	583.25
	1.75	366.2
	1.875	195.2375
	2	111.5
	2.125	82.32777778
	2.25	65.84
	2.375	58.7725
	2.5	55.7925
	2.625	53.59
	2.75	52.24
	2.875	51.635
	3	50.46
};
\addlegendentry{Serial BPL-SG-128}

\path[-] (rel axis cs:0,0)     coordinate(leftstart1)
          --(rel axis cs:0,0.65)coordinate(interruptleftA)
         (rel axis cs:0,0.725)  coordinate(leftstop1);

\path[-] (rel axis cs:0,0.725)     coordinate(leftstart2)
          --(rel axis cs:0,0.80)coordinate(interruptleftB1)
         (rel axis cs:0,0.8875)  coordinate(interruptleftB2)
         --(rel axis cs:0,1)   coordinate(leftstop2);

\path[-] (rel axis cs:1,0)     coordinate(rightstart)
         --(rel axis cs:1,1)   coordinate(rightstop);
\end{axis}


\draw(leftstart1)-- (interruptleftA) decorate[decoration={zigzag,segment length = 1.5mm}]{--(leftstop1)};
\draw(leftstart2)-- (interruptleftB1) decorate[decoration={zigzag,segment length = 1.5mm}]{--(interruptleftB2)} -- (leftstop2);
\draw[line width=1pt](rightstart) -- (rightstop);

\end{tikzpicture}
	\caption{\pf{Average number of iterations for parallel/serial architectures for UL-$(1024,512)$ polar codes with $I_{\max}=50$, where the parallel architecture refers to~\cite{elkelesh2018belief}.}}
	\label{fig:sec3_twoschedule}
\end{figure}
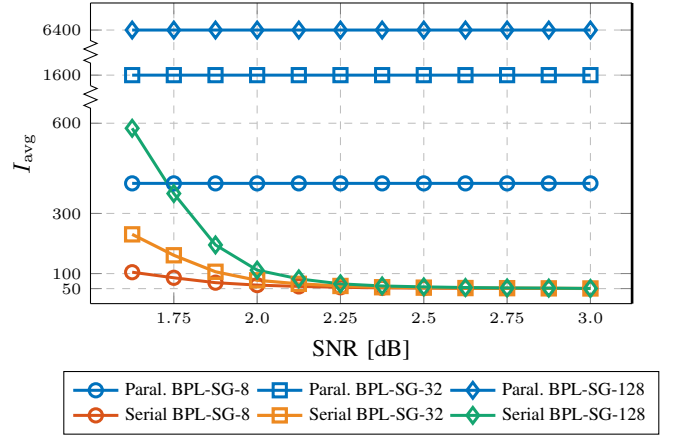

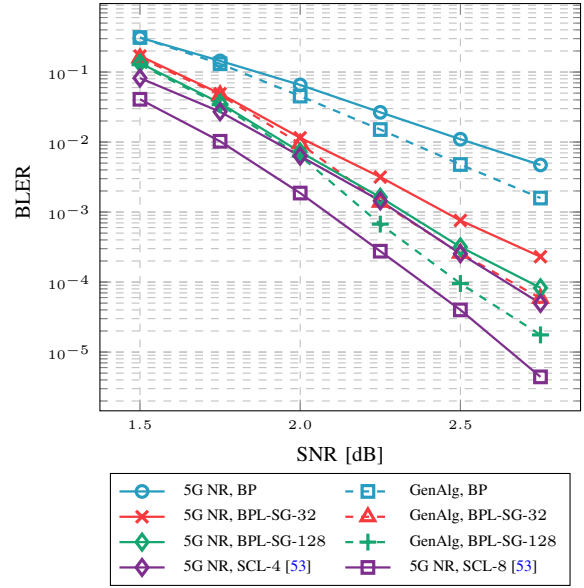
\begin{figure}[t]
	\centering
	\pgfplotsset{compat=1.16}
\begin{tikzpicture}
\definecolor{myskyblue}{RGB}{39,156,191}
\definecolor{myred}{RGB}{245, 57, 61}
\definecolor{mygreen}{RGB}{21, 165, 112}
\definecolor{myblued}{RGB}{0,114,189}
\definecolor{mypurple}{RGB}{126,47,142}
\definecolor{myyellow}{RGB}{243,165,93}
  \pgfplotsset{
    label style = {font=\fontsize{9pt}{7.2}\selectfont}, 
    tick label style = {font=\fontsize{9pt}{7.2}\selectfont} 
  }

\begin{axis}[
	scale = 1,
    ymode=log,
    xlabel={SNR [\text{dB}]},
    xlabel style={yshift=0.0cm}, 
    ylabel={BLER},
    ylabel style={xshift=0.2cm},
    grid=both,
    ymajorgrids=true,
    xmajorgrids=true,
    grid style=dashed,
    xmin = 1.375,
    xmax = 2.875,
    xtick={1.5,2.0,2.5,3.0},
    xticklabels={$1.5$,$2.0$,$2.5$,$3.0$},
    xticklabel style = {font=\tiny},
    yticklabel style = {font=\tiny},
    xlabel style = {font = \footnotesize},
    ylabel style = {font = \footnotesize},
    width=0.9\columnwidth, height=7cm,
    mark size=3.2,
    legend style={
      anchor={center},
      cells={anchor=west},
      column sep= 1.7mm, 
      font=\fontsize{6pt}{8}\selectfont, 
    },
    legend columns=2, 
    legend style={at={(0.5,-0.43)}, anchor=south},
]

\addplot[
    color=myskyblue,
    mark=o,
    mark options={scale=1,solid,very thick},
    very thick,
    line width=0.3mm,
    mark size=2.2,
]
table {
1.50000000000000	0.310000000000000
1.75000000000000	0.144333333333333
2	0.0654285714285714
2.25000000000000	0.0265853658536585
2.50000000000000	0.0109777777777778
2.75000000000000	0.00470879801734820
};
\addlegendentry{\text{5G NR, BP}}

\addplot[
    color=myskyblue,
    mark=square,
    mark options={scale=1,solid,very thick},
    very thick,
    dashed,
    line width=0.3mm,
    mark size=2.2,
]
table {
1.5      0.311550000000000
1.75     0.131250000000000
2.0      0.0453500000000000
2.25     0.0150888900000000
2.5      0.00477678570000000
2.75     0.00158459020000000
};
\addlegendentry{\text{GenAlg, BP}}

\addplot[
    color=myred,
    mark=x,
    mark options={scale=1,solid,very thick},
    very thick,
    line width=0.3mm,
    mark size=3.2,
]
table {
1.50000000000000	0.169500000000000
1.75000000000000	0.0491666666666667
2	0.0115000000000000
2.25000000000000	0.00316260162601626
2.50000000000000	0.000758730158730159
2.75000000000000	0.000230070218917803
};
\addlegendentry{\text{5G NR, BPL-SG-$32$}}

\addplot[
    color=myred,
    mark=triangle,
    mark options={scale=1,solid,very thick},
    very thick,
    dashed,
    line width=0.3mm,
    mark size=3,
]
table {
1.5      0.159550000000000
1.75     0.0471000000000000
2.0      0.00950000000000000
2.25     0.00138000000000000
2.5      0.000258333333333333
2.75     5.85714285714286e-05
};
\addlegendentry{\text{GenAlg, BPL-SG-$32$}}

\addplot[
color=mygreen,
mark=diamond,
mark options={scale=1,solid,very thick},
very thick,
line width=0.3mm,
mark size=3,
]
table {
1.50000000000000	0.134500000000000
1.75000000000000	0.0360000000000000
2	0.00721428571428571
2.25000000000000	0.00162601626016260
2.50000000000000	0.000317460317460317
2.75000000000000	8.26104915324246e-05
};
\addlegendentry{\text{5G NR, BPL-SG-$128$}}

\addplot[
    color=mygreen,
    mark=+,
    mark options={scale=1,solid,very thick},
    very thick,
    dashed,
    line width=0.3mm,
    mark size=3.2,
]
table {
1.5      0.127350000000000
1.75     0.0342000000000000
2.0      0.00635000000000000
2.25     0.000672222222222222
2.5      9.52380952380952e-05
2.75     1.75900000000000e-05
};
\addlegendentry{\text{GenAlg, BPL-SG-$128$}}

\addplot[
    color=mypurple,
    mark=diamond,
    mark options={scale=1,solid,very thick},
    very thick,
    line width=0.3mm,
    mark size=3,
]
table {
1.5      0.0820120300000000
1.75     0.0270880400000000
2.0      0.00629379400000000
2.25     0.00144530300000000
2.5      0.000253858900000000
2.75     4.95000000000000e-05
};
\addlegendentry{\text{5G NR, SCL-$4$ \cite{Stimming_2015}}}

\addplot[
    color=mypurple,
    mark=square,
    mark options={scale=1,solid,very thick},
    very thick,
    line width=0.3mm,
    mark size=2.2,
]
table {
1.5      0.0408163300000000
1.75     0.0102441500000000
2.0      0.00186773900000000
2.25     0.000275819500000000
2.5      4.00000000000000e-05
2.75     4.42000000000000e-06
};
\addlegendentry{\text{5G NR, SCL-$8$ \cite{Stimming_2015}}}

\end{axis}
\end{tikzpicture}%
	\caption{BLER performance of BP, the proposed BPL-SG, and SCL decoding for $(1024,512)$ polar codes with GenAlg and 5G NR constructions. All iterative decoders based on OMS decoding use CRC for detection, $I_{\max}=50$, and CRC-$11$ is specified by 5G NR.}
	\label{fig:FER_2}
\end{figure}
\begin{equation}\label{eq:CRC-11}
	\pf{g_{\mathrm{CRC-11}}(x)=x^{11}+x^{10}+x^{9}+x^{5}+1.}
\end{equation}

In addition to the 5G NR code construction that is unfriendly to BP decoding, we also consider the GenAlg construction~\cite{Ahmed2019Tcom} to fully show the potential of BP decoding.
Fig.~\ref{fig:FER_2} illustrates the BLER performance of BP, SCL, and BPL-SG decoding for $(1024,512)$ polar codes with the GenAlg and 5G NR constructions.
\pf{The CRC-$11$ polynomial from the 5G~NR standard, as shown in~\eqref{eq:CRC-11}, is adopted for both constructions.}
BPL-SG-$32$ under the GenAlg construction approaches the performance of SCL-$4$, and BPL-SG-$128$ under the GenAlg construction surpasses SCL-$4$ by $0.2$ dB at BLER~$=10^{-4}$.

\section{Proposed Algorithm of Formula-Based Permutation Generation}\label{sec:Theorem_drivation}
As mentioned in Section~\ref{sec:Sec2_D}, permutations of the factor graph stages can also be substituted efficiently by only shuffling the input LLRs to avoid instantiating multiple BP decoders with different factor graph architectures \pf{and} greatly facilitates the hardware implementation.
However, since a variety of shuffling patterns need to be generated to realize a single shuffling set $\mathcal{L}$ and since different code configurations require different shuffling sets, the implementation of the corresponding flexible LLR routing is \pf{challenging}.
In this section, we propose a hardware-friendly algorithm to generate these flexible routings \pf{(i.e., permutations)}.
We prove that the routing of any PFG can be decomposed into a combination of $n-1$ fixed sub-routings, as shown in Fig.~\ref{fig:FGP_process}.
The complicated hardware routing issue can \pf{thus} be optimized as a matrix decomposition.
\begin{figure}[t]
	\centering
	\includegraphics[width = 0.95\linewidth]{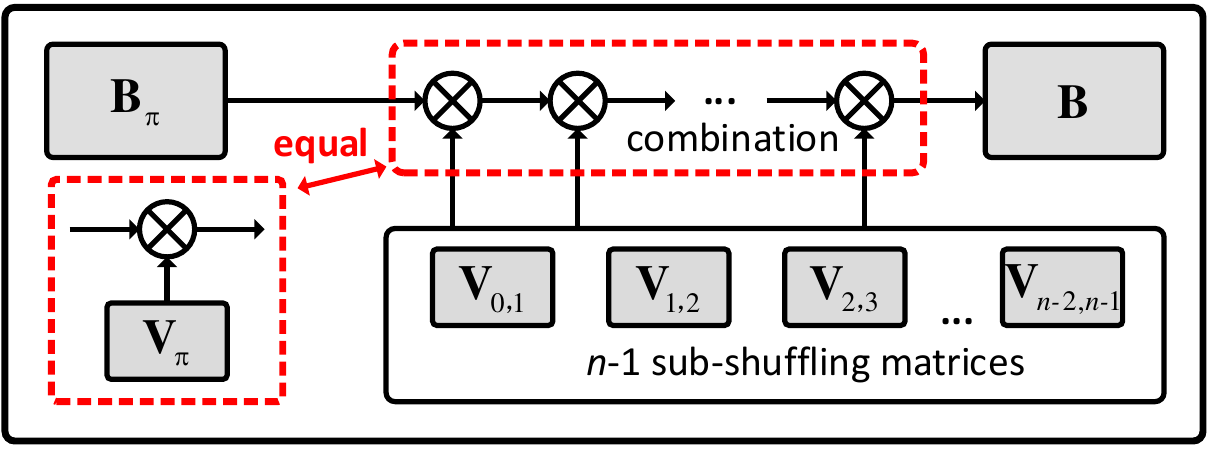}
	\caption{Model for permutations and the generation process of the shuffling matrix $\mathbf{V}_{\pi}$ for any PFG $\pi$.}
	\label{fig:FGP_process}
\end{figure}

\subsection{Mathematical Model for Permutations}
First, we create a model for permutations to solve the hardware routing problems with a matrix decomposition.
For length-$N$ polar codes, if the nodes $\bigoplus$ and $\boxed{=}$ in PFGs are represented by `$0$' and `$1$' respectively, any PFG can be mapped into a unique $N\times n$ binary matrix, in which the `$0/1$' sequence in each row corresponds to the binary expansion of the bit-index for the input LLRs.
For example, we use ${\mathbf{B}}=[\bm{b}_{n-1}\;\bm{b}_{n-2}\;\hdots\;\bm{b}_{1}\;\bm{b}_{0}]$ to denote the OFG of Fig.~\ref{fig:contr_8} (contrary to its stage order $[0\;1\;2\;\hdots\;n-1]$), and
$\bm{b}_{i}$ is a length-$N$ binary column vector expanded as
\begin{equation}\label{eq:m_i}
	\bm{b}_{i}=\biggl[\;\overbrace{\underbrace{\mathbf{0}}_{1\times2^{i}}\;\underbrace{\mathbf{1}}_{1\times2^{i}}\;\mathbf{0}\;\mathbf{1}\;\dots\;\mathbf{0}\;\mathbf{1}
	}^{2^{n-i-1}\;\mathrm{pairs}\;\mathrm{of}\;[{\mathbf{0}\;\mathbf{1}}]}\;\biggr]^\mathsf{T},\;i\in[0,n),
\end{equation}
where each $\mathbf{0}$ or $\mathbf{1}$ is an all-$0$ or all-$1$ row vector of length-$2^{i}$.
For any PFG $\pi$, the corresponding binary matrix ${\mathbf{B}}_{\pi}$ can be written as ${\mathbf{B}}_{\pi}=[\bm{b}_{\pi^{n-1}}\;\bm{b}_{\pi^{n-2}}\;\hdots\;\bm{b}_{\pi^{1}}\;\bm{b}_{\pi^{0}}]$, and~\eqref{eq:Sec4_B} shows the examples of $\OFG$ and $\pi_1$ of Fig.~\ref{fig:contr_8} and Fig.~\ref{fig:permutations}.
\begin{equation}\label{eq:Sec4_B}
	\mathbf{B}=\left[
	\begin{array}{c c c}
		0\;0\;0\\
		0\;0\;1\\
		0\;1\;0\\
		0\;1\;1\\
		1\;0\;0\\
		1\;0\;1\\
		1\;1\;0\\
		1\;1\;1\\
	\end{array}
	\right],
	\;\;\;\;\;
	\mathbf{B}_{\pi_1}=\left[
	\begin{array}{c c c}
		0\;0\;0\\
		0\;1\;0\\
		1\;0\;0\\
		1\;1\;0\\
		0\;0\;1\\
		0\;1\;1\\
		1\;0\;1\\
		1\;1\;1\\
	\end{array}
	\right].
\end{equation}

\begin{figure}[t]
	\centering
	\includegraphics[width=0.85\linewidth]{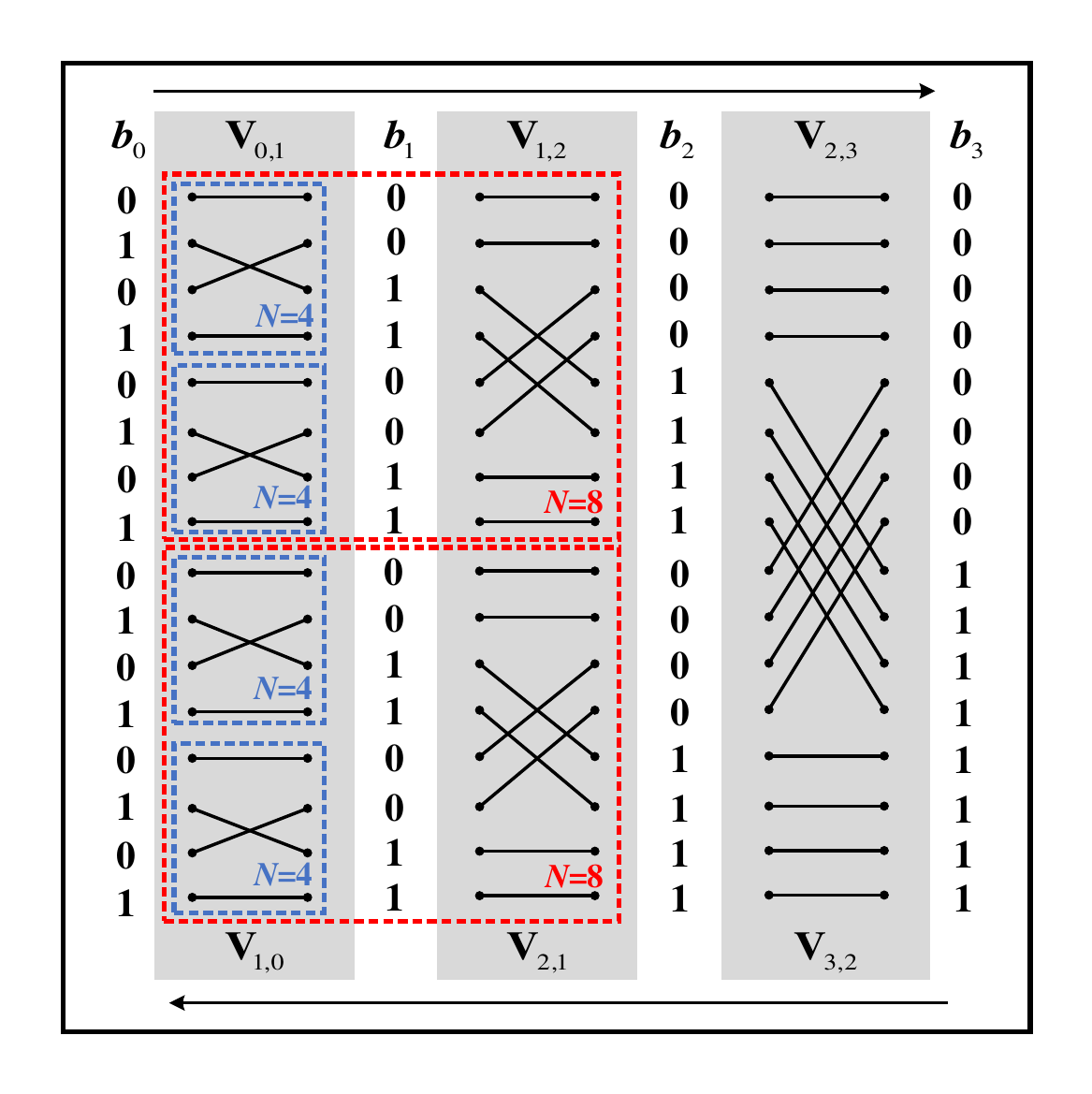}
	\caption{Sub-shuffling matrices and binary column vectors for length-$16$ codes.}\label{fig:sub_shuffling}
\end{figure}

Consequently, the model for permutations based on the input shuffling can be derived as follows
\begin{equation}\label{eq:bit_index}
	{\mathbf{B}}_{\pi}^\mathsf{T}\cdot{\mathbf{V}}_{\pi}={\mathbf{B}}^\mathsf{T}\rightarrow
	\left\{\begin{aligned}
		&\bm{b}_{\pi^{n-1}}^\mathsf{T}&&\cdot{\mathbf{V}}_{\pi} = \bm{b}_{n-1}^\mathsf{T},\\
		&\vdots\\
		&\bm{b}_{\pi^{1}}^\mathsf{T}&&\cdot{\mathbf{V}}_{\pi} = \bm{b}_{1}^\mathsf{T},\\
		&\bm{b}_{\pi^{0}}^\mathsf{T}&&\cdot{\mathbf{V}}_{\pi} = \bm{b}_{0}^\mathsf{T},\\
	\end{aligned}
	\right.
\end{equation}
where $\mathbf{V}_{\pi}$ is the shuffling matrix to represent the targeted routing in hardware.
Namely, as shown in Fig.~\ref{fig:FGP_process}, $\mathbf{B}_{\pi}$ can be multiplied by the corresponding $\mathbf{V}_{\pi}$ to obtain $\mathbf{B}$.
Hence, the above problem has been modelled as how to use a unified mathematical formula to express ${\mathbf{V}}_{\pi}$.

\subsection{Decomposition and Properties of ${\mathbf{V}}_{\pi}$}\label{sec:Formular_resp}
\begin{Thm}\label{thm:them_1}
	For any PFG $\pi$, the shuffling matrix ${\mathbf{V}}_{\pi}$ which satisfies ${\mathbf{B}}_{\pi}^\mathsf{T}\cdot{\mathbf{V}}_{\pi}={\mathbf{B}}^\mathsf{T}$ can be decomposed into a combination of $n-1$ fixed sub-shuffling matrices ${\mathbf{V}}_{i-1,i},i\in[1,n)$ for length-$N$ polar codes.
\end{Thm}

First, we provide the explicit expression for $n-1$ sub-shuffling matrices ${\mathbf{V}}_{i-1,i},i\in[1,n)$ for length-$N$ polar codes

\begin{equation}\label{eq:Vi-1i}
	\begin{aligned}
		&\bm{b}^\mathsf{T}\cdot{\mathbf{V}}_{i-1,i}\\
		&=[\overbrace{\underbrace{\bm{b}_{0}^{2^{i-1}-1}\;\bm{b}_{2^{i}}^{2^{i}+2^{i-1}-1}\;\bm{b}_{2^{i-1}}^{2^{i}-1}\;\bm{b}_{2^{i}+2^{i-1}}^{2^{i+1}-1}}_{2^{i+1}}\;\dots\;}^{2^{n}/2^{i+1}=2^{n-i-1}\;\mathrm{groups}}],
	\end{aligned}
\end{equation}
where ${\mathbf{V}}_{i-1,i}$ divides the input $\bm{b}$ into $2^{n-i-1}$ groups of equal length $2^{i+1}$ and shuffles based on 4 sub-vectors within each group.
For an intuitive understanding of Theorem~\ref{thm:them_1}, we use Fig.~\ref{fig:sub_shuffling} to illustrate a straightforward example of length-$16$ polar codes, which comprises $3$ sub-shuffling matrices $\{{\mathbf{V}}_{0,1}\;{\mathbf{V}}_{1,2}\;{\mathbf{V}}_{2,3}\}$ and $4$ binary column vectors $\{\bm{b}_{0}\;\bm{b}_{1}\;\bm{b}_{2}\;\bm{b}_{3}\}$ corresponding to~\eqref{eq:m_i} and~\eqref{eq:Vi-1i}.
Each sub-shuffling matrix describes a unique sub-routing in hardware.
Note that, due to the recursive construction of polar codes, $n-1$ sub-routings for length-$N$ polar codes can be decomposed into two independent copies of $n-2$ sub-routings for length-$\frac{N}{2}$ polar codes, as shown in Fig.~\ref{fig:sub_shuffling}.
Before the proof of Theorem~\ref{thm:them_1}, it is useful to introduce additional lemmas.
\begin{Lem}\label{lem:lem_1}
	$\forall$ $i\in[1,n)$, ${\mathbf{V}}_{i-1,i}$ is an involutory matrix, i.e., ${\mathbf{V}}_{i-1,i}={\mathbf{V}}_{i,i-1}$, ${\mathbf{V}}_{i-1,i}\cdot{\mathbf{V}}_{i,i-1}={\mathbf{I}}_{N}$.
\end{Lem}
Based on the features of~\eqref{eq:Vi-1i}, the proof is straightforward.
Then, we can further derive~\eqref{eq:lem1} as implied in Fig.~\ref{fig:sub_shuffling}
\begin{equation}\label{eq:lem1}
	\left\{
	\begin{aligned}
		&\bm{b}_{i-1}^\mathsf{T}\cdot{\mathbf{V}}_{i-1,i}&&=\bm{b}_{i}^\mathsf{T},\\
		&\bm{b}_{i}^\mathsf{T}\cdot{\mathbf{V}}_{i,i-1}&&=\bm{b}_{i-1}^\mathsf{T},\\
	\end{aligned}
	\right.
	\;\forall\;i\in[1,n).
\end{equation}

\begin{Lem}\label{lem:lem_2}
	\begin{equation}\label{eq:lem_2}
		\bm{b}_{i}^\mathsf{T}\cdot{\mathbf{V}}_{i,j}=\bm{b}_{j}^\mathsf{T},\;\forall\;i,j\in[0,n).
	\end{equation}
\end{Lem}

\begin{proof}
	We define ${\mathbf{V}}_{i,j}={\mathbf{V}}_{i,i+1}\cdot{\mathbf{V}}_{i+1,i+2}\dots{\mathbf{V}}_{j-1,j}$, $\forall\;i,j\in[0,n),\;i<j$ and let ${\mathbf{V}}_{i,j}={\mathbf{I}}_{N}$,  $\forall\;i,j\in[0,n),\;i=j$.
	Based on~\eqref{eq:lem1}, it can be verified that $\bm{b}_{i}^\mathsf{T}\cdot{\mathbf{V}}_{i,j}=\bm{b}_{i}^\mathsf{T}\cdot{\mathbf{V}}_{i,i+1}\cdot{\mathbf{V}}_{i+1,j}=\bm{b}_{i+1}^\mathsf{T}\cdot{\mathbf{V}}_{i+1,j}=\dots=\bm{b}_{j}^\mathsf{T}$.
	For $i\geq j$, a similar proof can be formulated.
\end{proof}

\begin{Lem}\label{lem:lem_3}
	$\forall\;i,j,k\in[0,n),i\neq j,k\neq i$,~\eqref{eq:lemma3} holds.
	\begin{equation}\label{eq:lemma3}
		\!\!\!\bm{b}_{k}^\mathsf{T}\cdot{\mathbf{V}}_{i,j}=\left\{\begin{aligned}
			&\bm{b}_{k}^\mathsf{T},&{\mathrm{if}}\;k\notin[\min(i,j),\max(i,j)],\\
			&\bm{b}_{k+\sgn(i-j)}^\mathsf{T},&{\mathrm{if}}\;k\in[\min(i,j),\max(i,j)].\\
		\end{aligned}
		\right.
	\end{equation}
\end{Lem}

An intuitive example for Lemma~\ref{lem:lem_3} is visible from Fig.~\ref{fig:sub_shuffling}:
Using Lemma~\ref{lem:lem_2}, when calculating $\bm{b}_{3}^\mathsf{T}\cdot\mathbf{V}_{0,2}\overset{\ref{lem:lem_2}}{=}\bm{b}_{3}^\mathsf{T}\cdot\mathbf{V}_{0,1}\cdot\mathbf{V}_{1,2}$, we can simply permute $\bm{b}_{3}^\mathsf{T}=[0\;\hdots\;0\;1\;\hdots\;1]^\mathsf{T}$ through $\mathbf{V}_{0,1}$ and $\mathbf{V}_{1,2}$ and find the result still equals to $\bm{b}_{3}^\mathsf{T}$ since $3\notin [0,2]$. Besides, using Lemma~\ref{lem:lem_1} and~\ref{lem:lem_2}, $\bm{b}_{1}^\mathsf{T}\cdot \mathbf{V}_{0,2}\overset{\ref{lem:lem_2}}{=}\bm{b}_{1}^\mathsf{T}\cdot \mathbf{V}_{0,1}\cdot\mathbf{V}_{1,2}\overset{~\ref{lem:lem_1},~\ref{lem:lem_2}}{=}\bm{b}_{0}^\mathsf{T}\cdot\mathbf{V}_{1,2}=\bm{b}_{0}^\mathsf{T}$ since $1\in [0,2]$ and $\sgn(0-2)=-1$.
\begin{proof}
	If $k\notin[\min(i,j),\max(i,j)]$, Lemma~\ref{lem:lem_3} can be verified by Lemma~\ref{lem:lem_1}, \eqref{eq:m_i}, and \eqref{eq:Vi-1i}.
	If $k\in[\min(i,j),\max(i,j)]$, we need to distinguish two cases: for $i<j$, $\bm{b}_{k}^\mathsf{T}\cdot{\mathbf{V}}_{i,j}$ can be represented by \eqref{eq:lemm3ij1} and therefore
	\begin{equation}\label{eq:lemm3ij1}
		\begin{aligned}
			\bm{b}_{k}^\mathsf{T}\cdot{\mathbf{V}}_{i,j}=&\bm{b}_{k}^\mathsf{T}\cdot{\mathbf{V}}_{i,i+1}\dots{\mathbf{V}}_{k-1,k}\cdot{\mathbf{V}}_{k,k+1}\dots{\mathbf{V}}_{j-1,j}\\
			=&\bm{b}_{k-1}^\mathsf{T}\cdot{\mathbf{V}}_{k,k+1}\dots{\mathbf{V}}_{j-1,j}=\bm{b}_{k-1}^\mathsf{T}.\\
		\end{aligned}
	\end{equation}

	For $i>j$, the proof for $\bm{b}_{k}^\mathsf{T}\cdot{\mathbf{V}}_{i,j}=\bm{b}_{k+1}^\mathsf{T}$ is similar.
\end{proof}

\begin{Lem}\label{lem:lem_4}
	For any ${\bf{B}}_{\pi}$, after the matrix multiplication by ${\mathbf{V}}_{i-1,i},\forall\;i\in[1,n)$, ${\bf{B}}_{\pi}$ never contains two identical columns.
\end{Lem}

\begin{proof}
	$\forall\;i\in[1,n)$, let $i-1=\pi^{x},i=\pi^{y}$, using Lemma~\ref{lem:lem_1} and Lemma~\ref{lem:lem_3}, it is apparent that
	\begin{equation}\label{eq:lemm4diff}
		\begin{aligned}
			&{\mathbf{B}}_{\pi}^\mathsf{T}\cdot{\mathbf{V}}_{i-1,i}\\
			=&[\bm{b}_{\pi^{n-1}}^\mathsf{T}\;\dots\;\bm{b}_{\pi^{x}}^\mathsf{T}\;\dots\;\bm{b}_{\pi^{y}}^\mathsf{T}\;\dots\;\bm{b}_{\pi^{0}}^\mathsf{T}]\cdot{\mathbf{V}}_{i-1,i}\\
			\overset{\ref{lem:lem_1},~\ref{lem:lem_3}}{=}&[\bm{b}_{\pi^{n-1}}^\mathsf{T}\;\dots\;\bm{b}_{\pi^{x}}^\mathsf{T}\cdot{\mathbf{V}}_{i-1,i}\;\dots\;\bm{b}_{\pi^{y}}^\mathsf{T}\cdot{\mathbf{V}}_{i-1,i}\;\dots\;\bm{b}_{\pi^{0}}^\mathsf{T}]\\
			=&[\bm{b}_{\pi^{n-1}}^\mathsf{T}\;\dots\;\bm{b}_{\pi^{y}}^\mathsf{T}\;\dots\;\bm{b}_{\pi^{x}}^\mathsf{T}\;\dots\;\bm{b}_{\pi^{0}}^\mathsf{T}].\\
		\end{aligned}
	\end{equation}
	Hence, for any $\pi$, the matrix multiplication by any ${\mathbf{V}}_{i-1,i},i\in[1,n)$ is equivalent to swapping two columns of ${\mathbf{B}}_{\pi}$.
\end{proof}

\begin{Lem}\label{lem:lem_5}
	For any ${\bf{B}}_{\pi}$, if executing a right-to-left column-wise transformation to realize ${\mathbf{B}}_{\pi}\rightarrow{\mathbf{B}}$, the previously matched columns are never influenced by the current shuffling matrix.
\end{Lem}

\begin{proof}
	$\forall i\in[0,n]$, there are $i$ matched columns on the right and $n-i$ un-matched columns on the left of ${\bf{B}}_{\pi}$ compared with $\bf{B}$, as shown in \eqref{eq:lamma5_1}

	\begin{equation}\label{eq:lamma5_1}
		\begin{aligned}
			&{\mathbf{B}}_{\pi_{s_i}}=[\underbrace{\bm{b}_{\pi_{s_i}^{n-1}}\;\dots\;\bm{b}_{\pi_{s_i}^{i}}}_{\rm{un-matched}}\;\underbrace{\bm{b}_{i-1}\;\dots\;\bm{b}_{0}}_{\mathrm{matched}}]\\
			&\rightarrow{\mathbf{B}}=[\bm{b}_{n-1}\;\dots\;\bm{b}_{i}\;\bm{b}_{i-1}\;\dots\;\bm{b}_{0}],\\
		\end{aligned}
	\end{equation}
	where $\pi_{s_i}$ denotes the original $\pi$ after $i$ transformations.
	Subsequently, we multiply the shuffling matrix ${\mathbf{V}}_{\pi_{s_i}^{i},i}$ to obtain~\eqref{eq:lamma5_2}
	\begin{equation}\label{eq:lamma5_2}
		\begin{aligned}
			&[\underbrace{\bm{b}_{\pi_{s_i}^{n-1}}^\mathsf{T}\;\dots\;\bm{b}_{\pi_{s_i}^{i}}^\mathsf{T}}_{\rm{un-matched}}\;\underbrace{\bm{b}_{i-1}^\mathsf{T}\;\dots\;\bm{b}_{0}^\mathsf{T}}_{\rm{matched}}]\cdot{\bf{V}}_{\pi_{s_i}^{i},i}\\
			=&[\{\underbrace{\bm{b}_{\pi_{s_i}^{n-1}}^\mathsf{T}\;\dots\;\bm{b}_{\pi_{s_i}^{i}}^\mathsf{T}}_{\rm{un-matched}}\}\cdot{\bf{V}}_{\pi_{s_i}^{i},i}\;\{\underbrace{\bm{b}_{i-1}^\mathsf{T}\;\dots\;\bm{b}_{0}^\mathsf{T}}_{\rm{matched}}\}\cdot{\bf{V}}_{\pi_{s_i}^{i},i}]\\
			=&[\underbrace{\bm{b}_{\pi_{s_{i+1}}^{n-1}}^\mathsf{T}\;\dots\;\bm{b}_{\pi_{s_{i+1}}^{i+1}}^\mathsf{T}}_{\rm{un-matched}}\;\underbrace{\bm{b}_{i}^\mathsf{T}\;\bm{b}_{i-1}^\mathsf{T}\;\dots\;\bm{b}_{0}^\mathsf{T}}_{\rm{matched}}].\\
		\end{aligned}
	\end{equation}

	In accordance with Lemma~\ref{lem:lem_4}, the un-matched columns in $\mathbf{B}_{\pi}$ never contain any element of $[\bm{b}_{i-1}^\mathsf{T}\;\dots\;\bm{b}_{0}^\mathsf{T}]$.
	Therefore, we can derive that $\pi_{s_i}^{i}>i-1$ and $[\bm{b}_{i-1}^\mathsf{T}\;\dots\;\bm{b}_{0}]\cdot{\mathbf{V}}_{\pi_{s_i}^{i},i}\overset{\ref{lem:lem_3}}{=}[\bm{b}_{i-1}^\mathsf{T}\;\dots\;\bm{b}_{0}^\mathsf{T}]$ always holds.
	Besides, combined with Lemma~\ref{lem:lem_1}, we further obtain $\bm{b}_{\pi_{s_{i}}^{i}}^\mathsf{T}\cdot\mathbf{V}_{\pi_{s_{i}}^{i},i}=\bm{b}_{i}^\mathsf{T}$ and let $\pi_{s+1}$ denote the original $\pi$ after $i+1$ transformations.
	Hence, the proof of Lemma~\ref{lem:lem_5} has been completed.
\end{proof}

In conclusion, combined with the aforementioned Lemma~\ref{lem:lem_1}-\ref{lem:lem_5}, the proof of Theorem~\ref{thm:them_1} is provided below.

\begin{proof}[Proof of Theorem 1]
	The decomposition process of any permutation is illustrated in~\eqref{eq:thm_proof1}, which clearly presents how to generate the shuffling matrix $\mathbf{V}_{\pi}$.
	\begin{equation}\label{eq:thm_proof1}
		\begin{aligned}
			&{\mathbf{B}}_{\pi}^\mathsf{T}\cdot{\mathbf{V}}_{\pi^{0},0}={\mathbf{B}}_{\pi_{s_1}}^\mathsf{T}=[\bm{b}_{\pi_{s_1}^{n-1}}^\mathsf{T}\;\dots\;\bm{b}_{\pi_{s_1}^{1}}^\mathsf{T}\;\bm{b}_{0}^\mathsf{T}],\\
			&{\mathbf{B}}_{\pi_{s_1}}^\mathsf{T}\cdot{\mathbf{V}}_{\pi_{s_1}^{1},1}={\mathbf{B}}_{\pi_{s_2}}^\mathsf{T}=[\bm{b}_{\pi_{s_2}^{n-1}}^\mathsf{T}\;\dots\;\bm{b}_{\pi_{s_2}^{2}}^\mathsf{T}\;\bm{b}_{1}^\mathsf{T}\;\bm{b}_{0}^\mathsf{T}],\\
			&\hdots\\
			&{\mathbf{B}}_{\pi_{s_{n-1}}}^\mathsf{T}\cdot{\mathbf{V}}_{\pi_{s_{n-1}}^{n-1},n-1}={\mathbf{B}}^\mathsf{T}=[\bm{b}_{n-1}^\mathsf{T}\;\dots\;\bm{b}_{1}^\mathsf{T}\;\bm{b}_{0}^\mathsf{T}],\\
		\end{aligned}
	\end{equation}
	where we rewrite $\pi^0$ as $\pi_{s_{0}}^0$ to obtain a unified mathematical notation and further simplify the process of~\eqref{eq:thm_proof1} by~\eqref{eq:thm_proof2}

	\begin{equation}\label{eq:thm_proof2}
		\begin{aligned}
			{\mathbf{B}}_{\pi}^\mathsf{T}\cdot{\mathbf{V}}_{\pi}=&{\mathbf{B}}_{\pi}^\mathsf{T}\cdot{\mathbf{V}}_{\pi_{s_0}^{0},0}\cdot{\mathbf{V}}_{\pi_{s_1}^{1},1}\hdots{\mathbf{V}}_{\pi_{s_{n-1}}^{n-1},n-1}\\
			=&{\mathbf{B}}_{\pi}^\mathsf{T}\cdot\prod\limits_{i=0}^{n-1}{\mathbf{V}}_{\pi_{s_i}^{i},i},\\
		\end{aligned}
	\end{equation}
	in which we obtain the final expression for $\mathbf{V}_{\pi}$ as
	\begin{equation}\label{eq:thm_proof3}
		\begin{aligned}
			\mathbf{V}_{\pi}=\prod\limits_{i=0}^{n-1}{\mathbf{V}}_{\pi_{s_i}^{i},i},
		\end{aligned}
	\end{equation}
	where ${\bf{V}}_{\pi_{s_i}^{i},i}$ can be expressed as a product from $n-1$ sub-shuffling matrices ${\bf{V}}_{i-1,i},\;i\in[1,n)$, based on Lemma~\ref{lem:lem_2}.
\end{proof}

\begin{algorithm}[t]
	\caption{\texttt{Generation of Permutations by A Matrix Decomposition}}\label{alg:decomposition}
	\KwIn{$\mathsf{R}_0$, PFG $=[\pi^0\;\pi^1\;\hdots\;\pi^{n-1}]$, and OFG~$=[0\;1\;\hdots\;n-1]$}
	\KwOut{$\mathsf{R}_0$}
	\tcp{$\!\!$initialization}
	\upshape{$\bm{s}\leftarrow\{{\mathbf{0}}\}$}; \tcp{$\!\!$store $\pi_{s_0}^{0}$, $\pi_{s_1}^{1}$, $\hdots$, $\pi_{s_{n-1}}^{n-1}$}
	\tcp{$\!\!n-1$ sub-shuffling matrices}
	\upshape{${\mathbf{V}}_{\mathrm{set}}=[{\mathbf{V}}_{0,1}\;{\mathbf{V}}_{1,2}\;\dots\;{\mathbf{V}}_{n-2,n-1}]$}\;
	\tcp{$\!\!$generate ${\mathbf{V}}_{\pi}$}
	\For{\upshape $i=0$ to $n-1$}
	{
		\upshape{$s=$PFG[$i$];} \tcp{$\!\!$current column $\pi_{s_i}^{i}$}
		\upshape{$e=$OFG[$i$];} \tcp{$\!\!$aimed column $i$}
		\tcp{$\!\!$update PFG by $\mathbf{V}_{s,e}$}
		\For{\upshape $j=i$ to $n-1$}
		{
			\upshape{$\mathrm{PFG}[j]\leftarrow\mytextsf{updateStage}(\mathrm{PFG}[j],s,e)$}\;
		}
		\upshape{$\bm{s}[i]=s$;} \tcp{store the current ${\pi_{s_{i}}^{i}}$}
	}
	\tcp{$\!\!$execute ${\mathbf{V}}_{\pi}$ to permute input LLRs}
	\For{\upshape $i=0$ to $n-1$}
	{
		\upshape{$\mathsf{R}_0\leftarrow\mytextsf{subRouting}(\mathsf{R}_0,{\mathbf{V}}_{\mathrm{set}},\bm{s}[i],\mathrm{OFG}[i])$}\;
	}
\end{algorithm}
\begin{algorithm}[t]
	\caption{\texttt{updateStage()}}\label{alg:updateFG}
	\KwIn{\upshape $\pi^{\mathrm{in}}$, $s$, $e$}
	\KwOut{\upshape $\pi^{\mathrm{out}}$}
	\If{\upshape $\pi^{\mathrm{in}}==s$}
	{
		\upshape{$\pi^{\mathrm{out}}=e$;} \tcp{$\!\!$Lemma~\ref{lem:lem_2}}
	}
	\ElseIf{\upshape ${\pi^{\mathrm{in}}\in[\min(s,e),\max(s,e)]}$ and $s\neq e$}
	{
		\upshape{$\pi^{\mathrm{out}}=\pi^{\mathrm{in}}+\sgn(s-e)$;} \tcp{$\!\!$Lemma~\ref{lem:lem_3}}
	}
	\Else
	{
		\upshape{$\pi^{\mathrm{out}}=\pi^{\mathrm{in}}$;} \tcp{$\!\!$keeps constant}
	}
\end{algorithm}
\begin{algorithm}[t]
	\caption{\texttt{subRouting()}}\label{alg:SubRouting}
	\KwIn{$\mathsf{R}_0$, ${\mathbf{V}}_{\mathrm{set}}$ $s$, $e$}
	\KwOut{${\mathbf{V}}_{\pi}$}
	\If{$s<e$}
	{
		\For{\upshape $i=s$; $i<=e-1$; $i++$}
		{
			\upshape{$\mathsf{R}_0^\mathsf{T}=\mathsf{R}_0^\mathsf{T}\cdot{\mathbf{V}}_{\mathrm{set}}[i]$}\;
		}
	}
	\ElseIf {$s>e$}
	{
		\For{\upshape $i=s-1$; $i>=e$; $i--$}
		{
			\upshape{$\mathsf{R}_0^\mathsf{T}=\mathsf{R}_0^\mathsf{T}\cdot{\mathbf{V}}_{\mathrm{set}}[i]$}\;
		}
	}
\end{algorithm}

\subsection{A Hardware-Friendly Algorithm by Matrix Decomposition}\label{sec:MF_V}
\pf{Note that the above derivation in Section~\ref{sec:Formular_resp} helps to interpret how to generate $\mathbf{V}_{\pi}$ from an algorithm perspective.
From a hardware perspective, the generation process is useful as it enables the BPL decoder to gradually permute the input LLRs into a specified order that is equivalent to multiplying by $\mathbf{V}_{\pi}$ of~\eqref{eq:thm_proof3}.}
An intuitive example is shown in Fig.~\ref{fig:RoutingCase}:
\pf{To generate $\pi_1=[m_2\;m_0\;m_1]$ of length-$8$ polar codes, we get $\mathbf{V}_{\pi_1}=\mathbf{V}_{1,2}\cdot\mathbf{V}_{0,1}$ based on~\eqref{eq:thm_proof3}.
Then, we permute $\bm{u}$ and $\bm{x}$ through two routings corresponding to an application of $\mathbf{V}_{1,2}$ followed by $\mathbf{V}_{0,1}$, which is equivalent to directly passing through the routing of Fig.~\ref{fig:permutations}.}
Herein, the generation of $\mathbf{V}_{\pi}$ and the process of shuffling the input LLRs are summarized~in~Algorithm~\ref{alg:decomposition}.

\pf{Algorithm~\ref{alg:decomposition} has two main phases: \texttt{generate ${\mathbf{V}}_{\pi}$} and \texttt{execute ${\mathbf{V}}_{\pi}$}.
To initialize ${\mathbf{V}}_{\mathrm{set}}$, $n-1$ sub-shuffling matrices ${\mathbf{V}}_{i-1,i},i\in[1,n)$ defined in~\eqref{eq:Vi-1i} are loaded to it.
The vector $\bm{s}$ stores the sequences of sub-shuffling indices that generate $\pi_{s_{i}}^{i}$ as in~\eqref{eq:thm_proof3} and are found during the decomposition of the shuffling.
To fill $\bm{s}$, we run a $\mytextsf{for}$ loop in lines $3-8$ of Algorithm~\ref{alg:decomposition} to perform the right-to-left column-wise transformation as shown in~\eqref{eq:thm_proof1}.
Performing the $i$-th loop is equivalent to updating the PFG by $\mathbf{V}_{\pi_{s_i}^{i},i}$ (i.e., ${\bf{B}}_{\pi_{s_{i}}}\rightarrow{\bf{B}}_{\pi_{s_{i+1}}}$ in~\eqref{eq:thm_proof1}).
The functions of Lemma~\ref{lem:lem_2} and Lemma~\ref{lem:lem_3} are implemented by the \mytextsf{updateStage()} function as shown in Algorithm~\ref{alg:updateFG}.
Note that for length-$N$ polar codes, we can fill $\bm{s}$ completely within $n$ steps to generate ${\mathbf{V}}_{\pi}$.
}

\begin{figure}[t]
	\centering
	\includegraphics[width=0.9\linewidth]{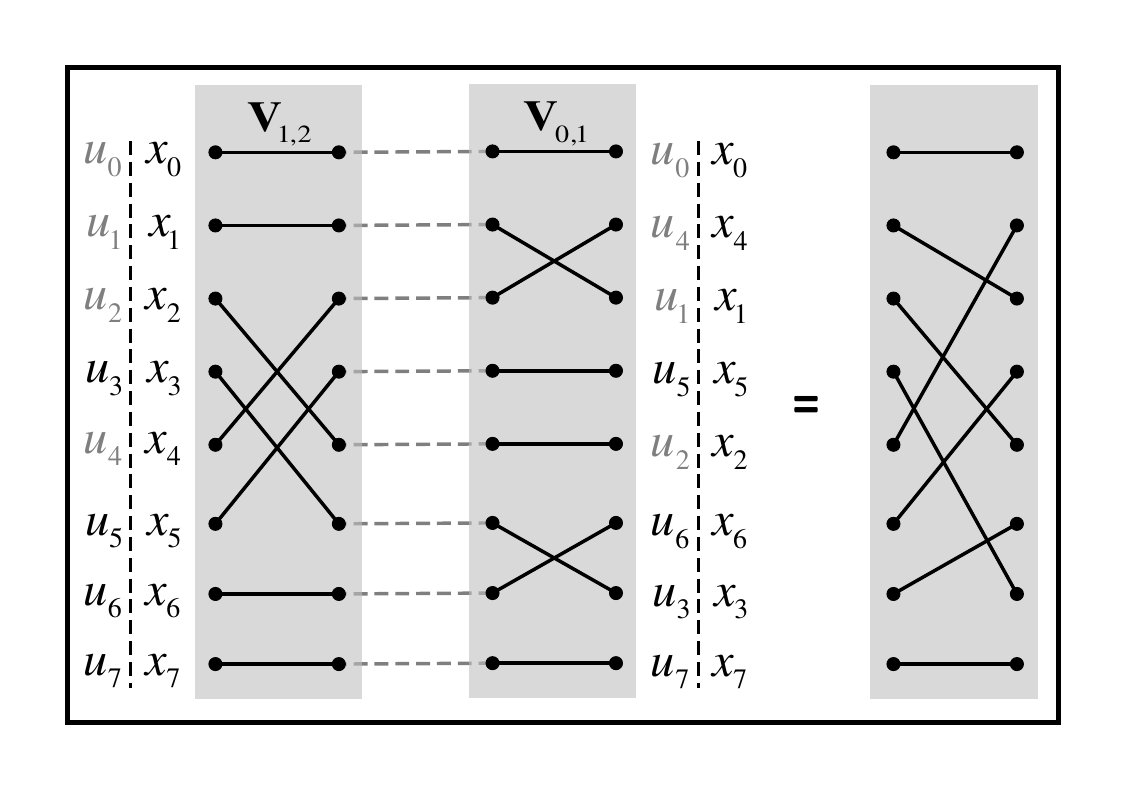}
	\caption{Shuffling the input LLRs based on a matrix decomposition for $\pi_1=[m_2\;m_0\;m_1]$.}\label{fig:RoutingCase}
\end{figure}

Subsequently, we run \mytextsf{subRouting()} as shown in Algorithm~\ref{alg:SubRouting} to permute the input LLRs $\mathsf{R}_0$ (or $\mathsf{L}_n$) based on the stored $\bm{s}$.
This function multiplies $\mathsf{R}_0$ by $\mathbf{V}_{\pi_{s_i}^{i},i}$ that can be decomposed into a product of $|\pi_{s_i}^{i}-i|$ sub-shuffling matrices.
Finally, we transmit the permuted $\mathsf{R}_0$ to the BP decoder based on the OFG.
\pf{The corresponding implementation is further explained in Section~\ref{subsec:ffgpg}.}

\begin{figure}[t]
	\centering
	\includegraphics[width = \linewidth]{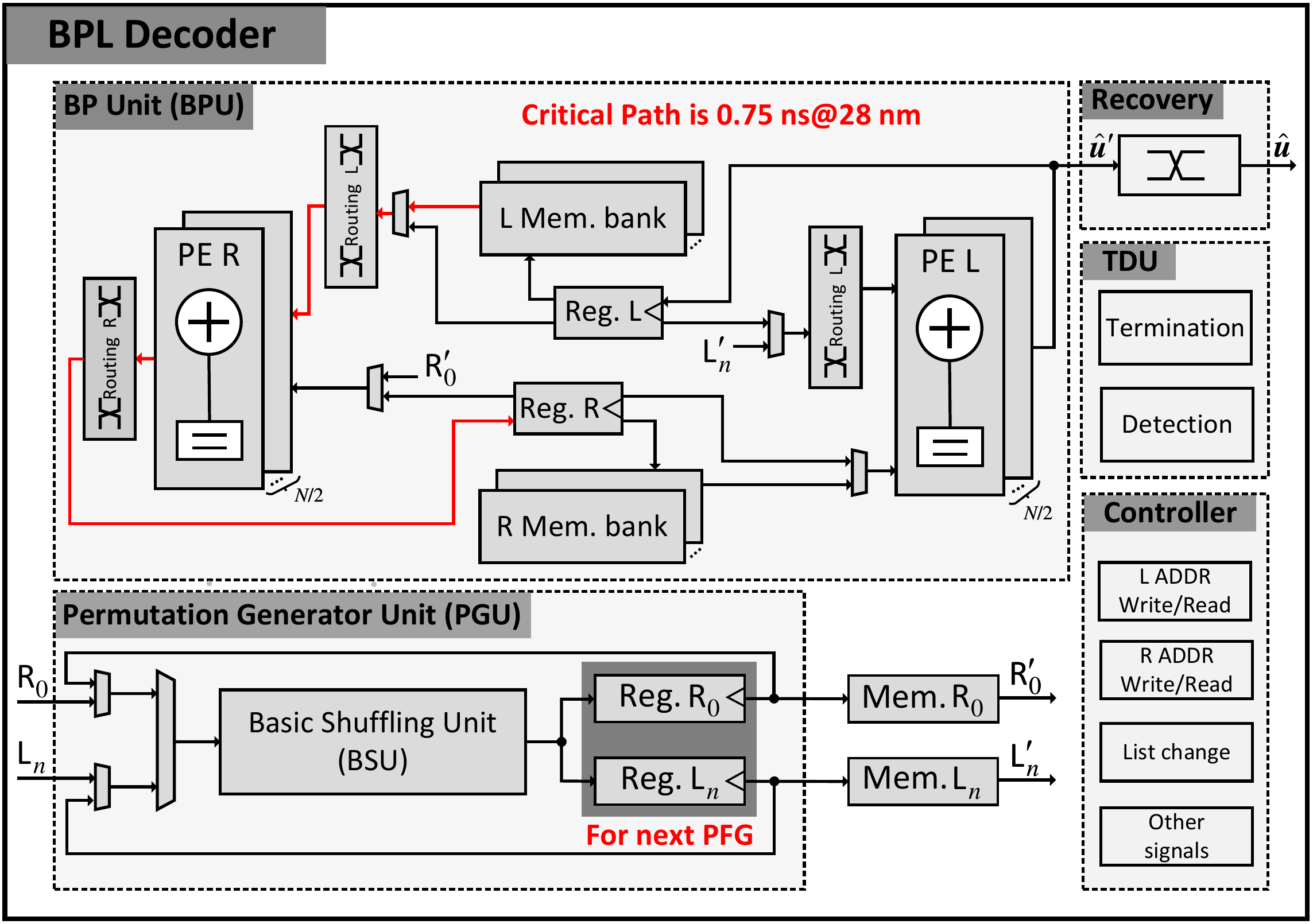}
	\caption{Hardware architecture of our BPL decoder, where solid black lines denote data signals and solid red lines denote the critical path.}
	\label{fig:BPL_topdesign}
\end{figure}

\section{Proposed BPL Decoder Architecture}\label{sec:bpl_flexible_fg_gen}
In this section, we present the architecture of our BPL decoder, which is the first hardware implementation of a BPL decoder for polar codes to the best of our knowledge.
Fig.~\ref{fig:BPL_topdesign} illustrates the overall hardware architecture, which comprises a double-column bidirectional-propagation BP unit (BPU)~\cite{chen2019}, a permutation generator unit (PGU) as proposed in the previous section to generate flexible permutations, a termination and detection unit (TDU), a recovery module, and a controller.

\subsection{Architecture Overview}\label{sec:bpl_decoder_overview}
For the underlying BPU, we employ the SOA double-column bidirectional-propagation architecture~\cite{chen2019}.
The dataflows of $\mathsf{R}$- and $\mathsf{L}$-messages are performed simultaneously, which means that $\mathsf{R}$-messages at the $j$-th stage and $\mathsf{L}$-messages at the $(n-j)$-th stage are calculated in the same clock cycle (CC).
We store the input LLRs $\mathsf{R}_{0}$ and $\mathsf{L}_{n}$ in the memory $\mathsf{R}_{0}$ and the memory $\mathsf{L}_{n}$, respectively.
The PGU shuffles the input LLRs according to Algorithm~\ref{alg:decomposition} to generate the shuffled input LLRs $\mathsf{R}'_{0}$ and $\mathsf{L}'_{n}$ for $\mathbb{L}$ PFGs.
Same as~\cite{chen2019}, the proposed decoder reduces the number of CCs per internal iteration from $n$ to $n-1$ by removing the calculation of $\mathsf{R}_{n}$ and $\mathsf{L}_{0}$ messages.
Besides, we use a sign-assisted (SA) termination strategy~\cite{Sun2016ISCAS,Ji2020TCAS1} to check the sign convergence of the internal results.
The SA strategy terminates decoding when hard decisions in the BPU are identical in three consecutive iterations, as shown in the termination module on the right side of Fig.~\ref{fig:BPL_topdesign}.
Finally, a detection module performs the CRC detection to judge whether output the current decoded $\bm{\hat{u}}$ or decode on a new PFG further, which is discussed in Section~\ref{sec:TD}.

\subsection{Permutation Generation Unit (PGU)}\label{subsec:ffgpg}
\subsubsection{Hardware Architecture}
Based on the shuffling matrix derived in Section~\ref{sec:Theorem_drivation}, we implement a low-complexity permutation network PGU for the BPL decoder, which comprises a basic shuffling unit (BSU), two registers for $\mathsf{R}_0$ and $\mathsf{L}_n$, and some MUXes.
For the BSU shown in Fig.~\ref{fig:BSU}, we instantiate $n-1$ fixed sub-routings ($\mathbf{V}_{i-1,i},i\in[1,n)$) to realize all the required basic shuffling.
All the stage orders of the proposed near-optimal PFG set $\mathcal{L}$ ($[\set_l^{0}\;\set_l^{1}\;\hdots\;\set_l^{n-1}],l\in[0,\mathbb{L})$ from the SG algorithm in Section~\ref{sec:mc_sec3}) is generated offline and stored in the PFG memory.
The inputs of the BSU are the $NQ$-bit initial input LLRs ($\mathsf{R}_{0}$ or $\mathsf{L}_{n}$, and each LLR is quantized as $Q$ bits) to be shuffled and the PFG index $l$ of the selected PFG.
First, the controller of the BSU uses $n$ CCs to obtain the set $\bm{s}$ of Algorithm~\ref{alg:decomposition} based on the output from the PFG memory, i.e., the stage orders of $\set_l$.
Subsequently, the BSU executes Algorithm~\ref{alg:SubRouting} step by step and controls the MUX to store the correct shuffled results into registers.
Note that, in order to make a trade-off between the hardware complexity and the latency of the permutation, we only perform one sub-routing per CC to shuffle the input LLRs.
For example, to realize $\mathbf{V}_{1,5}$, we decompose it as $\mathbf{V}_{1,2},\mathbf{V}_{2,3},\mathbf{V}_{3,4}$, and $\mathbf{V}_{4,5}$ to sequentially perform the desired permutations in $4$ CCs.
This decomposition process means that the latency for the permutation generation of any PFG is varying, as shown in~\eqref{eq:Sec5_latency_pi},
\begin{equation}\label{eq:Sec5_latency_pi}
	\mathfrak{L}_{\pi}=\sum_{i=0}^{n-1}|\pi^{i}_{s_{i}}-i| + n,
\end{equation}
where $\pi^{i}_{s_{i}}$ comes from~\eqref{eq:thm_proof3}.
The maximum latency of the BSU is $\mathfrak{L}_{\pi}=\frac{n\cdot(n-1)}{2}+n$ CCs when the PFG is $[n-1\;n-2\;\hdots\;1\;0]$.

\begin{figure}[t]
	\centering
	\includegraphics[width=0.925\linewidth]{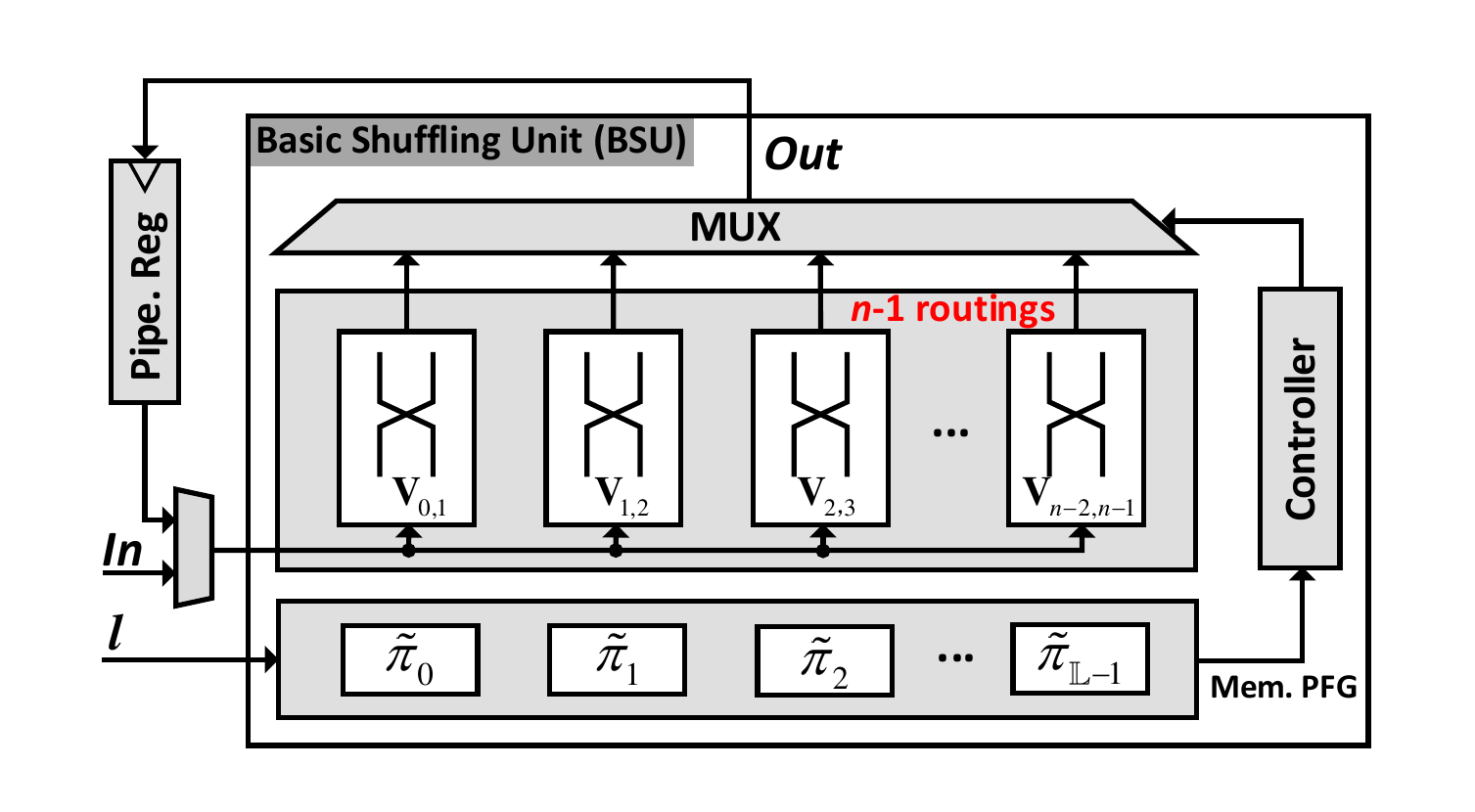}
	\caption{The core architecture of the permutation generation unit (PGU), based on the basic shuffling unit (BSU) and outside registers.}\label{fig:BSU}
\end{figure}

\subsubsection{Comparison With the Bene\v{s} Network}
\pf{To highlight the advantage and significance of our permutation network}, we reproduce \pf{a classical Bene\v{s} network~\cite{benevs1964optimal}} illustrated in Fig.~\ref{fig:benes} to make a fair comparison.
The number of inputs for a regular Bene\v{s} network is a power of two ($N=2^{n}$).
It has $2n-1$ stages, each with $N/2$ switches of size $2\times2$.
However, due to lack of an explicit method to generate control signals for PFGs in the Bene\v{s} network on-the-fly, for each $\set_{l}$ in $\mathcal{L}$, one would need to store $\frac{N\cdot(2n-1)}{2}$ bits in the control memory to configure the `BAR' or `CROSS' states of each $2\times2$ switch, as shown in dashed red lines in Fig.~\ref{fig:benes}.
It is obvious to see that the area of the control memory in the Bene\v{s} network linearly grows with the maximum list size.

\pf{Synthesis results show that in Table~\ref{tab:PGUvsBenes}, for length-$1024$ polar codes, the area overhead of our flexible permutation generator under different list sizes is only $0.076$ mm$^2$ using $28$~nm FD-SOI.}
This network \pf{can support and generate} an arbitrary number of PFG candidates without any area overhead.
Compared to the Bene\v{s} network,\footnote{\pf{The routing overhead of a Bene\v{s} network can be further reduced by a folded architecture, but this is beyond the scope of this paper. Moreover, the control memory occupies $75\%$ and $89\%$ of the area overhead of the Bene\v{s} network for $\mathbb{L}=8$ and $\mathbb{L}=32$, respectively.}} our work has \pf{an} $\{86.9\%\;96.3\%\}$ smaller area when \pf{$\mathbb{L}=8$ and $\mathbb{L}=32$}, respectively.
In terms of the permutation latency, the average $\mathfrak{L}_{\pi}$ of all $10!$ PFGs is $32.5$ CCs, and the maximum $\mathfrak{L}_{\pi}$ is $55$ CCs, which is higher than that of the Bene\v{s} network.
However, to alleviate this issue, we propose an optimized decoupled decoding schedule \pf{well-suited for} the serial architecture, which is discussed in detail in Section~\ref{sec:Sec5_para}.

\begin{figure}[t]
	\centering
	\includegraphics[width=0.98\linewidth]{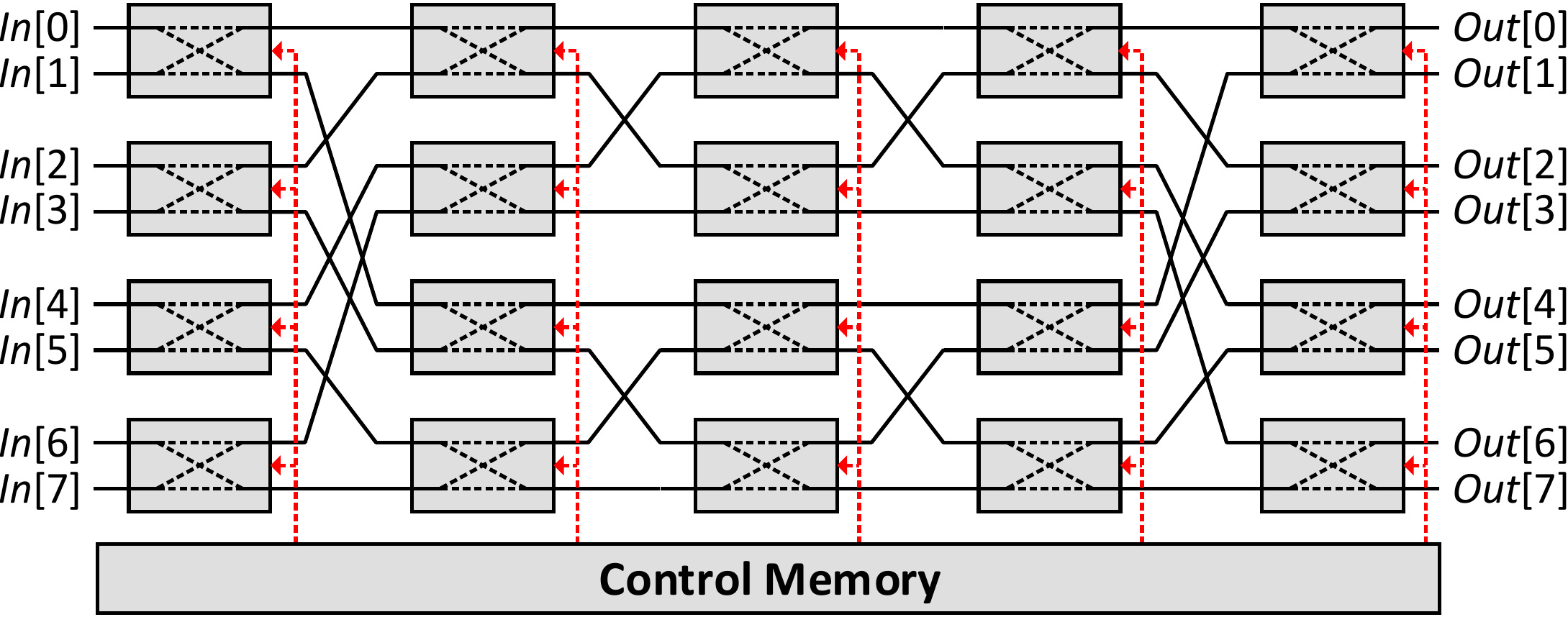}
	\caption{Architecture of a length-$8$ Bene\v{s} network~\cite{benevs1964optimal}.}\label{fig:benes}
\end{figure}

\begin{table}[t]
	\tabcolsep  4.2mm
	\renewcommand{\arraystretch}{1.1}
	\small
	\centering
	\caption{Synthesis results of different permutation networks for length-$1024$ polar codes in 28~nm FD-SOI with a 0.75~ns Target}
	\begin{center}
		{
			\begin{tabular}{lV{3}ccV{3}cc}
				\Xhline{1.2pt}
				& \multicolumn{2}{cV{3}}{This work}  & \multicolumn{2}{c}{Bene\v{s} work~\cite{benevs1964optimal}$^{\dag}$}
				\\ \Xhline{1.2pt}
				\textbf{List}                         & \multicolumn{1}{c|}{$8$} & $32$ & \multicolumn{1}{c|}{$8$}       & $32$
				\\ \hline
				\textbf{Area} {[}mm$^2${]}           & \multicolumn{2}{cV{3}}{$0.076$}    & \multicolumn{1}{c|}{$0.579$}   & $2.055$
				\\ \hline
				\textbf{Avg latency} {[}CC{]}       & \multicolumn{2}{cV{3}}{$32.5$}       & \multicolumn{2}{c}{\multirow{2}{*}{$1$}}
				\\ \cline{1-3}
				\textbf{Max latency} {[}CC{]} & \multicolumn{2}{cV{3}}{$55$}      & \multicolumn{2}{c}{}
				\\ \Xhline{1.2pt}
			\end{tabular}
			\label{tab:PGUvsBenes}
			\begin{tablenotes}
				\footnotesize
				\item[*] $^{\dag}$ \cite{benevs1964optimal} was re-implemented and synthesized \pf{using} $28$~nm FD-SOI.
			\end{tablenotes}
		}
	\end{center}
\end{table}

\subsection{Termination and Detection Module (TDU)}\label{sec:TD}
The TDU is integrated into the decoder and responsible for terminating the decoding once it has converged and for performing the CRC detection between BP decoding.
As described in Section~\ref{sec:bpl_decoder_overview}, we adopt the SA strategy \cite{Sun2016ISCAS,Ji2020TCAS1} to terminate the decoding if the hard decisions in the BPU are identical in three consecutive iterations.
Subsequently, if passing the termination module in Fig.~\ref{fig:BPL_topdesign} or achieving $I_{\max}$ (as the green SA shows in Fig.~\ref{fig:BPLTiming}), the BPL decoder requires 1 CC to calculate the final HD results $\bm{\hat{u}}'$.
Since the decoded codeword $\bm{\hat{u}}'$ is permuted, to perform the CRC detection, we transform $\bm{\hat{u}}'$ to $\bm{\hat{u}}$ with the natural order in the recovery module.
Note that this transformation is only the inverse process of the permutation generation in Section~\ref{subsec:ffgpg}, but the bit-width of this permutation network is only $N$-bit instead of $NQ$-bit.
The latency of recovering $\bm{\hat{u}}$ is $(\mathfrak{L}_{\pi}-n)$ CCs.
If the CRC detection succeeds, $\bm{\hat{u}}$ is output from our BPL decoder.

\subsection{Optimized Decoding Schedule}\label{sec:Sec5_para}
\begin{figure*}[t]
	\centering
	\includegraphics[width=\linewidth]{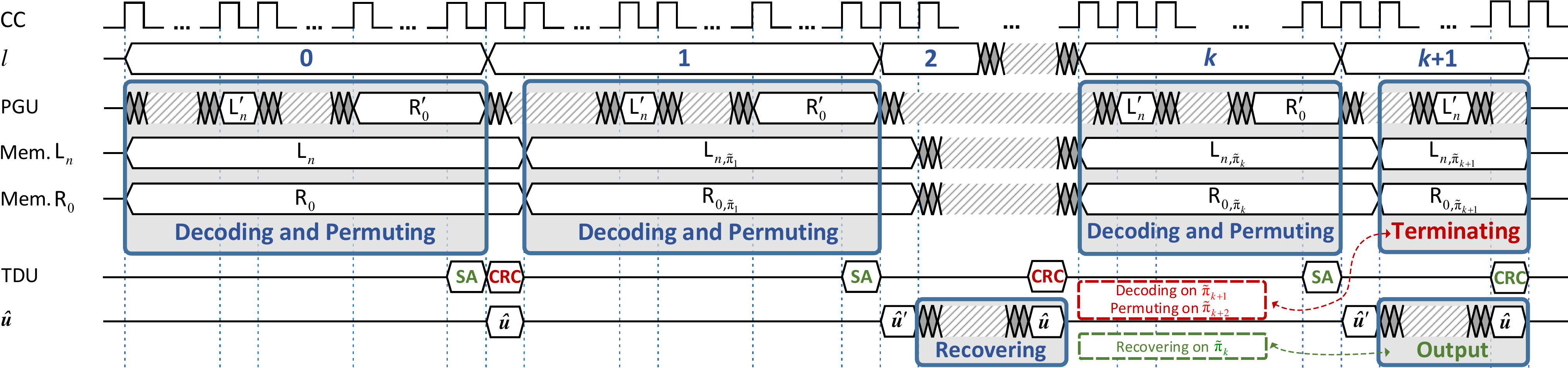}
	\caption{Timing schedule of the proposed BPL decoder, which decodes successfully on $\set_k$.}\label{fig:BPLTiming}
\end{figure*}

To improve the throughput of the BPL decoder, we need to avoid the influence of the latency from the proposed permutation network.
To this end, we propose an optimized decoding schedule which decouples the BPU, PGU, and the recovery modules with pipeline registers to allow them to work in parallel.
First, when the BPU works on $\set_l,l\in[0,\mathbb{L}-1)$, the PGU is activated to shuffle the input LLRs for the next PFG $\set_{l+1}$.
Since there are two kinds of the input LLRs in BP decoding ($\mathsf{R}_{0}$ and $\mathsf{L}_{n}$).
We use the module in Fig.~\ref{fig:BSU} to shuffle $\mathsf{R}_{0}$ and $\mathsf{L}_{n}$ one after another in the same hardware.
The detailed timing schedule of the proposed BPL decoder is illustrated in Fig.~\ref{fig:BPLTiming}.
To distinguish the shuffled input LLRs for different PFGs, we denote $\mathsf{R}_{0,\set_l}$ and $\mathsf{L}_{n,\set_l}$ as the permuted input signals $\mathsf{R}'_{0}$ and $\mathsf{L}'_{n}$ for $\set_l$.
When the BPL decoder performs BP decoding on $\set_l,l\in[0,\mathbb{L}-1)$, after the PGU has already generated $\mathsf{R}_{0,\set_{l+1}}$ and $\mathsf{L}_{n,\set_{l+1}}$, these two signals are temporarily stored into the register $\mathsf{R}_{0}$ and the register $\mathsf{L}_{n}$ that are marked with a dark grey background in Fig.~\ref{fig:BPL_topdesign}.
Once BP decoding on $\set_l$ passes the termination module, the memory $\mathsf{R}_{0}$ and the memory $\mathsf{L}_{n}$ are updated by the register $\mathsf{R}_{0}$ and the register $\mathsf{L}_{n}$ to output $\mathsf{R}_{0,\set_{l+1}}$ and $\mathsf{L}_{n,\set_{l+1}}$ as $\mathsf{R}'_{0}$ and $\mathsf{L}'_{n}$ for the next BP decoding on $\set_{l+1}$.

Moreover, as said in Section~\ref{sec:TD}, the recovery from $\bm{\hat{u}}'$ to $\bm{\hat{u}}$ also harms the throughput of the BPL decoder, since the inverse permutation operations come at the cost of $\mathfrak{L}_{\pi}-n$~CCs.
To deal with this issue, we further decouple the recovery module from the decoding schedule.
This decoupled decoding schedule allows the BPU, the PGU, and the recovery modules to work on $\{\set_{l+1}\;\set_l\;\set_{l-1}\},l\in[1,L-1)$, respectively, which significantly improves the throughput of the BPL decoder and enhances the hardware utilization.
For example, we assume that the BPL decoder has successfully decoded on $\set_k$ in Fig.~\ref{fig:BPLTiming}.
When passing the CRC detection on $\set_k$, the BPL decoder terminates the decoding on $\set_{k+1}$ and the permutation generation on $\set_{k+2}$.

\section{Implementation Results}\label{sec:implementation}
In this section, we present the synthesis results for our BPL implementation.
All synthesis results are based on $28$~nm FD-SOI technology in the typical-typical corner, and we use timing constraints that are not achievable to maximum the operating frequency for our design.
A comparison with the SOA polar decoders is also provided.

\subsection{Quantized Performance}\label{sec:Sec6_quantized}
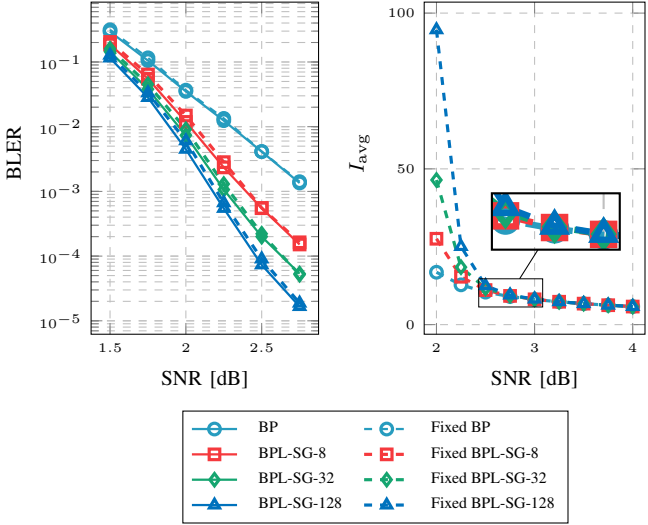
\begin{figure}[t]
	\centering
	\usetikzlibrary{calc,arrows, shapes, positioning, patterns, patterns.meta, snakes, pgfplots.groupplots,spy}

\usepgfplotslibrary{groupplots}
\pgfplotsset{compat=1.16}
\begin{tikzpicture}[spy using outlines]

	\definecolor{myskyblue}{RGB}{39,156,191}
	\definecolor{myred}{RGB}{245, 57, 61}
	\definecolor{mygreen}{RGB}{21, 165, 112}
	\definecolor{myblued}{RGB}{243,165,93} 
	\definecolor{mypurple}{RGB}{0,114,189} 
	\definecolor{myblack}{RGB}{126,47,142} 
	\pgfplotsset{
		label style = {font=\fontsize{9pt}{7.2}\selectfont}, 
		tick label style = {font=\fontsize{9pt}{7.2}\selectfont} 
	}

	\usetikzlibrary{
		matrix,
	}
	\begin{axis}[
		scale = 1,
		ymode=log,
		xlabel={SNR [\text{dB}]},
		xlabel style={yshift=0.0cm}, 
		ylabel={BLER},
		ylabel style={xshift=0.2cm},
		grid=both,
		ymajorgrids=true,
		xmajorgrids=true,
		grid style=dashed,
		xmin = 1.375,
		xmax = 2.875,
		xticklabel style = {font=\tiny},
		yticklabel style = {font=\tiny},
		xlabel style = {font = \footnotesize},
		ylabel style = {font = \footnotesize},
		xshift=-0.5\columnwidth,
		width=0.52\linewidth, height=6cm,
		mark size=2.2,
		legend style={
			anchor={center},
			cells={anchor=west},
			column sep= 2mm, 
			font=\fontsize{6pt}{8}\selectfont, 
		},
		legend columns=2, 
		legend style={at={(1.25,-0.57)}, anchor=south},
		]

	\addplot[
		color=myskyblue,
		mark=o,
		mark options={scale=1,solid,very thick},
		very thick,
		line width=0.3mm,
		mark size=2.2,
		]
		table {
		1.500000e+00 2.940000e-01
		1.750000e+00 1.050000e-01
		2 3.510000e-02
		2.250000e+00 1.240417e-02
		2.500000e+00 4.116463e-03
		2.750000e+00 1.349531e-03
		};
	\addlegendentry{\text{BP}}

	\addplot[
		color=myskyblue,
		mark=o,
		mark options={scale=1,solid,very thick},
		very thick,
		dashed,
		line width=0.3mm,
		mark size=2.2,
		]
		table {
		1.500000e+00 3.142000e-01
		1.750000e+00 1.159500e-01
		2 3.670000e-02
		2.250000e+00 1.342778e-02
		2.500000e+00 4.136567e-03
		2.750000e+00 1.412418e-03
		};
	\addlegendentry{\text{Fixed BP}}

	\addplot[
		color=myred,
		mark=square,
		mark options={scale=1,solid,very thick},
		very thick,
		line width=0.3mm,
		mark size=2,
	]
	table {
		1.500000e+00 1.865000e-01
		1.750000e+00 5.510000e-02
		2 1.185000e-02
		2.250000e+00 2.350000e-03
		2.500000e+00 5.500000e-04
		2.750000e+00 1.511905e-04
	};
	\addlegendentry{\text{BPL-SG-8}}

	\addplot[
		color=myred,
		mark=square,
		dashed,
		line width=0.3mm,
		mark size=2,
		mark options={scale=1,solid,very thick},
		very thick,
	]
	table {
		1.500000e+00 2.041000e-01
		1.750000e+00 6.405000e-02
		2 1.490000e-02
		2.250000e+00 2.833333e-03
		2.500000e+00 5.545455e-04
		2.750000e+00 1.635135e-04
	};
	\addlegendentry{\text{Fixed BPL-SG-8}}

	\addplot[
		color=mygreen,
		mark=diamond,
		mark options={scale=1,solid,very thick},
		very thick,
		line width=0.3mm,
		mark size=2.32,
	]
	table {
		1.500000e+00 1.440500e-01
		1.750000e+00 3.800000e-02
				2 7.350000e-03
		2.250000e+00 1.066667e-03
		2.500000e+00 2.000000e-04
		2.750000e+00 5.347826e-05
	};
	\addlegendentry{\text{BPL-SG-32}}

	\addplot[
		color=mygreen,
		mark=diamond,
		dashed,
		line width=0.3mm,
		mark size=2.3,
	    mark options={scale=1,solid,very thick},
	    very thick,
	]
	table {
		1.500000e+00 1.595500e-01
		1.750000e+00 4.515000e-02
		2 9.100000e-03
		2.250000e+00 1.300000e-03
		2.500000e+00 2.214286e-04
		2.750000e+00 5.169492e-05
	};
	\addlegendentry{\text{Fixed BPL-SG-32}}



	\addplot[
		color=mypurple,
		mark=triangle,
		mark options={scale=1,solid,very thick},
		very thick,
		line width=0.3mm,
		mark size=2.2,
	]
	table {
		1.500000e+00 1.152000e-01
		1.750000e+00 2.825000e-02
		2 4.475000e-03
		2.250000e+00 5.458333e-04
		2.500000e+00 7.317073e-05
		2.750000e+00 1.660784e-05
	};
	\addlegendentry{\text{BPL-SG-128}}

	\addplot[
		color=mypurple,
		mark=triangle,
		dashed,
		line width=0.3mm,
		mark size=2.2,
		mark options={scale=1,solid,very thick},
		very thick,
	]
	table {
		1.500000e+00 1.273500e-01
		1.750000e+00 3.330000e-02
		2 6.050000e-03
		2.250000e+00 6.666667e-04
		2.500000e+00 9.029851e-05
		2.750000e+00 1.875000e-05
		};
	\addlegendentry{\text{Fixed BPL-SG-128}}
	\end{axis}

	\begin{axis}[
 		scale = 1,
		xlabel={SNR [\text{dB}]},
		xlabel style={yshift=0.0cm}, 
		ylabel={$I_{\mathrm{avg}}$},
		ylabel style={xshift=0.2cm},
		grid=both,
		ymajorgrids=true,
		xmajorgrids=true,
		grid style=dashed,
		xmin = 1.875,
		xmax = 4.125,
		xticklabel style = {font=\tiny},
		yticklabel style = {font=\tiny},
		xlabel style = {font = \footnotesize},
		ylabel style = {font = \footnotesize},
		width=0.51\linewidth, height=6cm,
		mark size=3.2,
		]

	\addplot[
		color=myskyblue,
		mark=o,
		dashed,
		line width=0.3mm,
		mark size=2.2,
		mark options={scale=1,solid,very thick},
		very thick,
		]
	table {
		2 1.681969e+01
		2.250000e+00 1.277504e+01
		2.500000e+00 1.042848e+01
		2.750000e+00 9.024070e+00
		3 8.058145e+00
		3.250000e+00 7.328750e+00
		3.500000e+00 6.734750e+00
		3.750000e+00 6.236230e+00
		4 5.804815e+00
	};

	\addplot[
		color=myred,
		mark=square,
		dashed,
		line width=0.3mm,
		mark size=2,
		mark options={scale=1,solid,very thick},
		very thick,
	]
	table {
		2 2.756343e+01
		2.250000e+00 1.536767e+01
		2.500000e+00 1.109617e+01
		2.750000e+00 9.220130e+00
		3 8.122475e+00
		3.250000e+00 7.345125e+00
		3.500000e+00 6.739735e+00
		3.750000e+00 6.236620e+00
		4 5.804950e+00
	};

	\addplot[
		color=mygreen,
		mark=diamond,
		dashed,
		line width=0.3mm,
		mark size=2.3,
		mark options={scale=1,solid,very thick},
		very thick,
	]
	table {
		2 4.643039e+01
		2.250000e+00 1.844113e+01
		2.500000e+00 1.167778e+01
		2.750000e+00 9.328035e+00
		3 8.156660e+00
		3.250000e+00 7.346140e+00
		3.500000e+00 6.740035e+00
		3.750000e+00 6.236620e+00
		4 5.804950e+00
	};


	\addplot[
		color=mypurple,
		mark=triangle,
		dashed,
		line width=0.3mm,
		mark size=2.2,
		mark options={scale=1,solid,very thick},
		very thick,
	]
	table {
		2 9.463774e+01
		2.250000e+00 2.507674e+01
		2.500000e+00 1.248745e+01
		2.750000e+00 9.490565e+00
		3 8.193695e+00
		3.250000e+00 7.346140e+00
		3.500000e+00 6.740035e+00
		3.750000e+00 6.236620e+00
		4 5.804950e+00
	};

\end{axis}

\spy[rectangle, draw, width=1.7cm, height=0.75cm,magnification=2, connect spies] on (1.15,0.55) in node at (1.75,1.5);

\end{tikzpicture}%
	\caption{BLER comparison between floating-point and fixed-point and $I_{\mathrm{avg}}$ of the BPL decoder equipped with the proposed near-optimal set for $(1024,512)$ polar codes, $\mathbb{L}\in\{8,32,128\}$, and $I_{\max}=50$.}
	\label{fig:Quan_3}
\end{figure}

In Fig.~\ref{fig:Quan_3}, under the GenAlg construction, we present the BLER performance and $I_{\mathrm{avg}}$ of our BPL decoder using floating-point and fixed-point (2's complement).
Let $Q_{q_i.q_f}$ denote a fixed-point number with one sign-bit, $q_i-q_f-1$ integer bits, and $q_f$ fractional bits.
We adopt $Q_{7.2}$ for the LLRs in Fig.~\ref{fig:Quan_3}.
Numerical results show that $Q_{7.2}$ almost approaches the floating-point performance of OMS polar decoding with $[\beta_{\mathsf{R}}\;\beta_{\mathsf{L}}]=[0.25\;0]$ providing a well-balanced trade-off between the error-correcting performance and hardware complexity.
For the $I_{\avg}$ of our BPL decoder, the influence of the list size is mainly reflected in the low SNR regions, which achieves $47.8$ and $98.5$ iterations for $\mathbb{L}\in\{32\;128\}$ at SNR~$=2.0$~dB, respectively.
However, as SNR increases, the $I_{\avg}$ of the BPL decoder converges rapidly to that of the BP decoder, which achieves $I_{\avg}=5.81$ for $\mathbb{L}\in\{32\;128\}$.

\subsection{Latency Analysis}\label{sec:Sec6_latency}
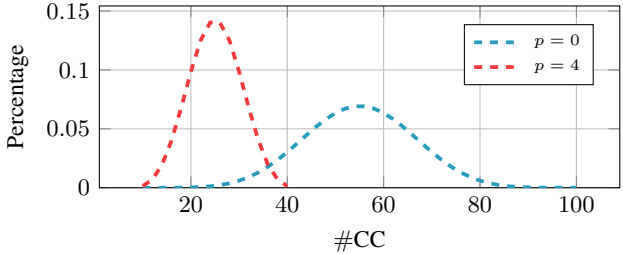
\begin{figure}[t]
	\centering
\pgfplotsset{compat=1.16}
\begin{tikzpicture}
\definecolor{mygrey}{RGB}{169,169,169}
\definecolor{myskyblue}{RGB}{39,156,191}
\definecolor{myred}{RGB}{245, 57, 61}
\definecolor{mygreen}{RGB}{21, 165, 112}
\definecolor{myblued}{RGB}{243,165,93} 
\definecolor{mypurple}{RGB}{0,114,189} 
\definecolor{myblack}{RGB}{126,47,142} 
\pgfplotsset{
	label style = {font=\fontsize{9pt}{7.2}\selectfont}, 
	tick label style = {font=\fontsize{9pt}{7.2}\selectfont} 
}
	\begin{axis} [
		scale = 1,
		grid=both,
    	yminorgrids=true,
		xminorgrids=true,
		height=4cm,
		width=8.5cm,
		ymin=0,
		xlabel = {$\#$CC},
		ylabel = {Percentage},
		legend style={
			anchor={center},
			cells={anchor=west},
			column sep= 1mm, 
			font=\fontsize{6pt}{8}\selectfont, 
		},
		legend columns=1, 
		legend style={at={(0.95,0.9)}, anchor=north east},
    yticklabel style={
        /pgf/number format/fixed,
        /pgf/number format/precision=5
    },
    scaled y ticks=false
		]
		\addplot[
			color=myskyblue,
			dashed,
			line width=0.5mm,
		]
		coordinates {
			(10,2.755732e-07)
			(12,2.480159e-06)
			(14,1.212522e-05)
			(16,4.271384e-05)
			(18,1.212522e-04)
			(20,2.943122e-04)
			(22,6.332672e-04)
			(24,1.237048e-03)
			(26,2.230765e-03)
			(28,3.758818e-03)
			(30,5.971671e-03)
			(32,9.006559e-03)
			(34,1.296379e-02)
			(36,1.788167e-02)
			(38,2.371418e-02)
			(40,3.031581e-02)
			(42,3.743744e-02)
			(44,4.473573e-02)
			(46,5.179646e-02)
			(48,5.817047e-02)
			(50,6.341793e-02)
			(52,6.715553e-02)
			(54,6.909970e-02)
			(56,6.909970e-02)
			(58,6.715553e-02)
			(60,6.341793e-02)
			(62,5.817047e-02)
			(64,5.179646e-02)
			(66,4.473573e-02)
			(68,3.743744e-02)
			(70,3.031581e-02)
			(72,2.371418e-02)
			(74,1.788167e-02)
			(76,1.296379e-02)
			(78,9.006559e-03)
			(80,5.971671e-03)
			(82,3.758818e-03)
			(84,2.230765e-03)
			(86,1.237048e-03)
			(88,6.332672e-04)
			(90,2.943122e-04)
			(92,1.212522e-04)
			(94,4.271384e-05)
			(96,1.212522e-05)
			(98,2.480159e-06)
			(100,2.755732e-07)
		};
		\addlegendentry{$p=0$}
		\addplot[
			color=myred,
			dashed,
			line width=0.5mm,
		]
		coordinates {
			(10,1.388889e-03)
			(12,6.944444e-03)
			(14,1.944444e-02)
			(16,4.027778e-02)
			(18,6.805556e-02)
			(20,9.861111e-02)
			(22,1.250000e-01)
			(24,1.402778e-01)
			(26,1.402778e-01)
			(28,1.250000e-01)
			(30,9.861111e-02)
			(32,6.805556e-02)
			(34,4.027778e-02)
			(36,1.944444e-02)
			(38,6.944444e-03)
			(40,1.388889e-03)
		};
		\addlegendentry{$p=4$}
	\end{axis}
\end{tikzpicture}
	\caption{Latency distribution for permuting $\mathsf{R}_{0}$ or $\mathsf{L}_{n}$ of all PFGs for length-$1024$ polar codes, based on the proposed permutation network.}
	\label{fig:latency_bar}
\end{figure}

Fig.~\ref{fig:latency_bar} shows the latency distribution for permuting $\mathsf{R}_{0}$ and $\mathsf{L}_{n}$ of all PFGs for length-$1024$ polar codes (i.e., $p=0$) based on the PGU.
For any PFG, the latency of the PGU is denoted by $\mathfrak{L}'_{\pi}$ in~\eqref{eq:Sec6_latency_BPL_pi}
\begin{equation}\label{eq:Sec6_latency_BPL_pi}
	\mathfrak{L}'_{\pi}=2\mathfrak{L}_{\pi}-n=2\sum_{i=0}^{n-1}|\pi^{i}_{s_{i}}-i| + n,
\end{equation}
where the multiplication with $2$ is due to the reuse of the BSU module in Fig.~\ref{fig:BPL_topdesign} to shuffle $\mathsf{R}_{0}$ and $\mathsf{L}_{n}$.
This latency distribution presents an approximately normal distribution trend: the minimum is $10$~CCs and the maximum is $100$~CCs.
Note that, for all $10!$ PFGs, $\mathfrak{L}'_{\pi}$ of $96\%$ PFGs is lower than $80$ CCs.
Subsequently, when combined with the aforementioned BPL-SG algorithm, if we set $p=4$ in~\eqref{eq:Sec3_k} (left stages fixed as $[m_{0}\;m_{1}\;m_{2}\;m_{3}]$) to efficiently decrease the search space of PFGs, the dynamic range of the $\mathfrak{L}'_{\pi}$ distribution rapidly narrows, i.e., the maximum $\mathfrak{L}'_{\pi}$ has a $60\%$ reduction from $100$~CCs to $40$~CCs, as shown in Fig.~\ref{fig:latency_bar}.

\begin{figure*}[b]
	\hrulefill
	\begin{equation}\label{eq:Sec6:laten}
		\begin{aligned}
			\mathfrak{L}_{\rm{BPL}}=\sum_{l=0}^{k}\max\left(\underbrace{(n-1)\cdot 	I_{\set_l},}_{\mathrm{decoding\;on\;\set_{\emph{l}}}}\underbrace{\mathfrak{L}'_{\set_{l+1}},}_{\mathrm{permutation\;on\;\set_{\emph{l}+1}}}\underbrace{\mathfrak{L}_{\set_{l-1}}-n}_{\mathrm{recovery\;on\;\set_{\emph{l}-1}}}\right)+\underbrace{k+1}_{\mathrm{calculate\;}\bm{\hat{u}}'}+\underbrace{\mathfrak{L}_{\set_k}-n}_{\mathrm{recovery\;on\;\set_{\emph{k}}}},k\in[0,\mathbb{L}).\\
		\end{aligned}
	\end{equation}
\end{figure*}

Consequently, the whole decoding latency of our BPL decoder is calculated by~\eqref{eq:Sec6:laten}, where $I_{\set_l}$ denotes a practical iteration number of $\set_l$ and $k$ is the index of the first PFG that delivers a successfully decoded codeword.
Note that, to satisfy the uniformity of~\eqref{eq:Sec6:laten}, we set $\mathfrak{L}'_{\set_{\mathbb{L}}}=\varnothing$, $\mathfrak{L}_{\set_{-1}}=\varnothing$, and $\mathfrak{L}_{\OFG}=0$.
The $\max$ operation represents how the proposed decoupled decoding schedule alleviates the influence of the permutation latency and of the recovery latency.
If we keep $I_{\max}$ relatively large, the decoding latency on $\set_l$ is always dominating in the $\max$ term.
Numerical results show that when $I_{\max}\geq15$,~\eqref{eq:Sec6:laten} can be simplified as
\begin{equation}\label{eq:Sec6_appr_laten}
	\mathfrak{L}_{\mathrm{BPL}}\approx\sum_{l=0}^{k}\left(\left(n-1\right)\cdot I_{\set_l}\right)+k+1+\mathfrak{L}_{\set_k}-n,
\end{equation}
where $k\in[0,\mathbb{L})$.
The average latency of our BPL decoder for $(1024,512)$ polar codes with $\mathbb{L}=32$ and $I_{\max}=50$ is only $53.25$ CCs at SNR~$=4$~dB.

\subsection{Comparisons With Previous Works}\label{sec:Sec6_comp}
\begin{table*}[t]
	\tabcolsep 0.28mm
	\renewcommand{\arraystretch}{1.1}
	\footnotesize
	\centering
	\caption{comparison with the soa polar decoders for $(1024,512)$ polar codes}
	\begin{center}
			{
					\begin{tabular}{l V{3} c | c | c | c V{3} c | c | c | c | c | c | c | c | c }
							\Xhline{1.2pt}
							\multirow{2}{*}{\textbf{Decoders}} & \multicolumn{4}{cV{3}}{\multirow{2}{*}{\textbf{This work}}} & [SSCL'22] & [TCAS-I'19] & [TVLSI'17] & [TCAS-II'20] & [TCOM'20] & [TSP'17] & [TCAS-I'20] & \pf{[TSP'20]} & [TSP'22]\\
							& \multicolumn{4}{cV{3}}{} & \cite{Su2022JSSCL} & \cite{chen2019} & \cite{Abbas17High} & \cite{Shen2020Tcas} & \cite{shen2020improved} & \cite{FastSSCL2017} & \cite{Eran2020} & \cite{Lee2020TSP} & \cite{ren2022sequence}\\
							\Xhline{1.2pt}
							\textbf{Algorithm} &  \multicolumn{4}{cV{3}}{\textbf{BPL}$^{\ddag}$}  &  \textbf{BP}   &  \textbf{BP}   &   \textbf{BP}  &  \textbf{EBPF}  &  \textbf{GBPF-MS}  &  \textbf{Fast-SSCL}  &   \textbf{Fast-SSCF} & \pf{\textbf{Fast-SCLF}} & \textbf{SR-List} \\
							\hline
							\textbf{Process} [nm] & \multicolumn{3}{c|}{$28$} & $\pf{65}$ & $40$  & $40$ & $65$  & $65$ & $40$ & $65$ & $65$ & \pf{$90$} & $28$\\
							\hline
							\textbf{Quantization} [bit] & \multicolumn{3}{c|}{$7$} & \multirow{4}{*}{$\pf{*}$} & $6$ & $5$  & $5$ & $7$  & $6$ & $6$ & $6$ & \pf{$6$} & $6$\\
							\cline{1-4} \cline{6-14}
							\textbf{List/Attempt} &  $8$  &  $32$ & $128$ & & $-$  & $-$ & $-$ & $20$  & $10$ & $4$ & $20$ & \pf{$4/128$} & $4$\\
							\cline{1-4} \cline{6-14}
							SNR$@$BLER~$=10^{-4}$ & $2.85$ & $2.65$ & $2.50$ & & \pf{$-$} & \pf{$3.9$} & \pf{$3.75$}  & \pf{$3.25$} & \pf{$2.74$} & $2.65$ & \pf{$3.00$}  & \pf{$2.50$} & $2.65$\\
							\cline{1-4} \cline{6-14}
							\textbf{Avg Iter./Attempt} &  \multicolumn{3}{c|}{$5.81^{\dag}$} & &  \pf{$4.4^{\dag}$}  & \pf{$7.36^{\dag}$} &  \pf{$6.34^{\dag}$}  & \pf{$4.34^{\dag}$} & \pf{$7.17$}  & $-$  & $1.01$  & \pf{$1.01$} & $-$\\
							\Xhline{1.2pt}
							\textbf{Area} [mm$^2$] & \multicolumn{3}{c|}{$0.87$} & $\pf{2.39}$ & $2.07$ & $0.704$ & $1.60$ & $3.11$ & $0.946$ & $1.822$ & $0.56$ & \pf{$2.83$} & $0.286$\\
							\hline
							\textbf{Frequency} [MHz] & \multicolumn{3}{c|}{$1333$} & $\pf{352}$ & $150$ & $500$ & $334$ & $319$ & $806$ & $840$ & $455$ & \pf{$615$} & $1255$\\
							\hline
							\textbf{Worst-Case T/P}$^{\star}$ [Gb/s] & $0.37$ & $0.09$ & $0.02$ & $\pf{*}$ & \pf{$0.16$} & \pf{$1.12$} & \pf{$1.36$} & \pf{$0.02$} & \pf{$0.06$} & $1.61$ & $0.076$ & \pf{$0.012$} & $3.62$\\
							\hline
							\textbf{Coded T/P} [Gb/s] & \multicolumn{3}{c|}{$25.63^{\dag}$} & $\pf{6.76^{\dag}}$ & \pf{$1.85^{\dag}$} & \pf{$7.61^{\dag}$} & \pf{$10.7^{\dag}$} & \pf{$3.72^{\dag}$} & \pf{$4.19$} & $1.61$ & $1.51$ & \pf{$1.52$} & $3.62$\\
							\hline
							\textbf{Area Eff.} [Gbps/mm$^2$] & \multicolumn{3}{c|}{$29.46^{\dag}$} & $\pf{2.83^{\dag}}$ & $0.894^{\dag}$ & $10.81^{\dag}$ & $6.687^{\dag}$ & $1.20^{\dag}$ & $4.43$ & $0.883$ & $2.71$ & $0.535$ & $12.67$\\
							\Xhline{1.2pt}
							\multicolumn{9}{l}{\pf{Normalized to $65$ nm, $1.0$ V$^{\S}$}}\\
							\Xhline{1.2pt}
							\textbf{Coded T/P} [Gbps] & \multicolumn{4}{cV{3}}{$\pf{6.76}$} & \pf{$1.14$} & \pf{$4.68$} & \pf{$10.7$} & \pf{$3.72$} & \pf{$2.58$} & \pf{$1.61$} & \pf{$1.51$} & \pf{$2.10$} & \pf{$1.56$}\\
							\hline
							\textbf{Area Eff.} [Gbps/mm$^2$] & \multicolumn{4}{cV{3}}{$\pf{2.83}$} & \pf{$0.208$} & \pf{$2.519$} & \pf{$6.687$} & \pf{$1.20$} & \pf{$1.032$} & \pf{$0.883$} & \pf{$2.71$} & \pf{$1.420$} & \pf{$1.012$}\\
							\Xhline{1.2pt}
						\end{tabular}\label{tab:tab_comp}}
			\begin{tablenotes}
					\footnotesize
					\item[*] $^{\dag}$ Average results reported at SNR $=4.0$ dB.
					\item[*] $^{\star}$ \pf{Worst-case results estimated at $I_{\max}=50$ for BP-based decoders.}
					\item[*] $^{\ddag}$ This work employs the double-column bidirectional-propagation architecture.
          \item[*] $^{*}$ \pf{Identical to the parameters for the $28$~nm results.}
					\item[*] $^{\S}$ Normalized to $65$ nm technology: area $\varpropto$ $s^2$ and frequency $\varpropto$ 1/$s$, \pf{where $s$ is the scaling factor to $65$ nm}.
				\end{tablenotes}
		\vspace{-0.1cm}
		\end{center}
\end{table*}

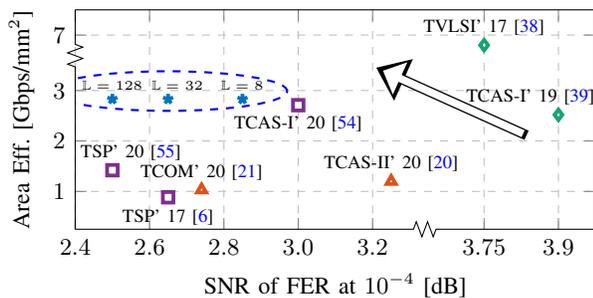
\begin{figure}[t]
	\centering
\pgfplotsset{compat=1.16}
\usetikzlibrary{arrows.meta, bending}
\begin{tikzpicture}
\usepgflibrary{decorations.pathmorphing}
\usepgflibrary[decorations.pathmorphing]
\usetikzlibrary{decorations.pathmorphing}
\usetikzlibrary[decorations.pathmorphing]

\usetikzlibrary{arrows}
\usetikzlibrary{shapes}

\definecolor{myblued}{RGB}{0,114,189}
\definecolor{myred}{RGB}{217,83,25}
\definecolor{myredd}{RGB}{248,74,173}
\definecolor{myyellow}{RGB}{237,137,32}
\definecolor{mypurple}{RGB}{126,47,142}
\definecolor{myblues}{RGB}{77,190,238}
\definecolor{mygreen}{RGB}{21, 165, 112}
\pgfplotsset{
	label style = {font=\fontsize{9pt}{7.2}\selectfont}, 
	tick label style = {font=\fontsize{9pt}{7.2}\selectfont} 
}
	\begin{axis} [
		scale = 1,
		grid=both,
		height=4.5cm,
		width=8.5cm,
		ymin=0.25,
		ymax=4.6,
		xmin=2.4,
		xmax=3.8,
		xlabel = {SNR of FER at $10^{-4}$ [dB]},
		ylabel = {Area Eff. [Gbps/mm$^2$]},
		xtick={2.4, 2.6, 2.8, 3.0, 3.2, 3.5, 3.7},
		xticklabels={$2.4$, $2.6$, $2.8$, $3.0$, $3.2$, $3.75$, $3.9$},
		ytick={1, 2, 3, 4.1},
		yticklabels={$1$, $2$, $3$, $7$},
		ymajorgrids=true,
		xmajorgrids=true,
		grid style=dashed,
		separate axis lines,
		x axis line style= { draw opacity=0 },
		y axis line style= { draw opacity=0 },
        thick
		]
		\addplot[
			color=myblued,
			very thick,
			mark = star,
			line width=0.5mm,
		]
		coordinates {
			(2.85,2.83)
		};
	
		\addplot[
			color=myblued,
			very thick,
			mark = star,
			line width=0.5mm,
		]
		coordinates {
			(2.65,2.83)
		};
	
		\addplot[
			color=myblued,
			very thick,
			mark = star,
			line width=0.5mm,
		]
		coordinates {
			(2.50,2.83)
		};

		\addplot[
			color=myred,
			very thick,
			mark = triangle,
			line width=0.5mm,
		]
		coordinates {
			(3.25,1.20)
		};
	
		\addplot[
			color=myred,
			very thick,
			mark = triangle,
			line width=0.5mm,
		]
		coordinates {
			(2.74,1.032)
		};
	
		\addplot[
			color=mypurple,
			very thick,
			mark = square,
			line width=0.5mm,
		]
		coordinates {
			(2.65, 0.883)
		};
	
		\addplot[
			color=mypurple,
			very thick,
			mark = square,
			line width=0.5mm,
		]
		coordinates {
			(2.50, 1.420)
		};
	
		\addplot[
			color=mypurple,
			very thick,
			mark = square,
			line width=0.5mm,
		]
		coordinates {
			(3.00, 2.71)
		};
	
		\addplot[
			color=mygreen,
			very thick,
			mark = diamond,
			line width=0.5mm,
		]
		coordinates {
			(3.7, 2.519)
		};
	
		\addplot[
			color=mygreen,
			very thick,
			mark = diamond,
			line width=0.5mm,
		]
		coordinates {
			(3.5, 3.9)
		};

		\node[black, above] at (axis cs:2.5, 2.83){\tiny{$\mathbb{L}=128$}};
		\node[black, above] at (axis cs:2.675, 2.83){\tiny{$\mathbb{L}=32$}};
		\node[black, above] at (axis cs:2.85, 2.83){\tiny{$\mathbb{L}=8$}};
		\node[draw=blue, dashed, ellipse, minimum width=90pt, minimum height=15pt, align=center] at (axis cs:2.65, 2.99){};
		
		\node[black, above] at (axis cs:3.25,1.20){\scriptsize{TCAS-II' 20~\cite{Shen2020Tcas}}};
		\node[black, above] at (axis cs:2.74,1.032){\scriptsize{TCOM' 20~\cite{shen2020improved}}};	
		
		\node[black, above] at (axis cs:2.65, 0.15){\scriptsize{TSP' 17~\cite{FastSSCL2017}}};
		\node[black, above] at (axis cs:2.55, 1.420){\scriptsize{TSP' 20~\cite{Lee2020TSP}}};	
		\node[black, above] at (axis cs:3.00, 2.0){\scriptsize{TCAS-I' 20~\cite{Eran2020}}};
	
		\node[black, above] at (axis cs:3.625, 2.519){\scriptsize{TCAS-I' 19~\cite{chen2019}}};
		\node[black, above] at (axis cs:3.5, 3.9){\scriptsize{TVLSI' 17~\cite{Abbas17High}}};
		
		\path[-] (rel axis cs:0,0)     coordinate(botstart)
		--(rel axis cs:0.65,0)coordinate(interruptbotA)
		(rel axis cs:0.70,0)  coordinate(interruptbotB)
		--(rel axis cs:1,0)   coordinate(botstop);

		\path[-] (rel axis cs:0,1)     coordinate(topstart)
		--(rel axis cs:1,1)   coordinate(topstop);
		
		\path[-] (rel axis cs:0,0)     coordinate(leftstart)
		--(rel axis cs:0,0.74)coordinate(interruptleftA)
		(rel axis cs:0,0.865)  coordinate(interruptleftB)
		--(rel axis cs:0,1)   coordinate(leftstop);
		
		\path[-] (rel axis cs:1,0)     coordinate(rightstart)
		--(rel axis cs:1,1)   coordinate(rightstop);
		
	\end{axis}
\draw(botstart)-- (interruptbotA) decorate[decoration={zigzag,segment length = 1.5mm, amplitude = 1mm}]{--(interruptbotB)} -- (botstop);
\draw(topstart)-- (topstop);
\draw(leftstart)-- (interruptleftA) decorate[decoration={zigzag,segment length = 1.5mm, amplitude = 1mm}]{--(interruptleftB)} -- (leftstop);
\draw(rightstart)-- (rightstop);

\path[draw=black,   line width=1.35mm, -{Triangle[length=4mm]}]    
(6,1.25) to    (4,2.15);
\path[draw=white, line width=1mm, -{Triangle[length=2.5mm]}, shorten >=1mm, shorten <=0.5mm]    
(6,1.25) to   (4,2.15);
\end{tikzpicture}
	\caption{\pf{SNR of FER at $10^{-4}$ vs. area efficiency (normalized to $65$~nm) for various SOA decoders for polar codes.}}
	\label{fig:performace_comp}
\end{figure}

In Table~\ref{tab:tab_comp}, we present the implementation results of the proposed BPL decoder using $28$~nm FD-SOI and compare them with the SOA architectures in~\cite{Abbas17High, chen2019, Su2022JSSCL, Shen2020Tcas, shen2020improved, FastSSCL2017, Eran2020}.
\pf{To ensure a fair comparison, we also implement our work based on $65$~nm CMOS technology since most published polar decoders use this process.}
For $\mathbb{L}=32$, our decoder with a near-optimal PFG set can achieve the error-correcting performance of BLER~$=10^{-4}$ at SNR~$=2.65$~dB, which is similar to that of Fast-SSCL-$4$ in~\cite{FastSSCL2017} and better than other BP and BPF works.
In terms of implementation results, our work has no advantage in terms of worst-case throughput, which is a common problem for a serial architecture.
However, our BPL decoder has an average throughput of $25.63$ Gbps,
\pf{which is $1.44\times$, $1.82\times$, $4.48\times$, and $4.33\times$ higher than the SOA BP~\cite{chen2019}, BPF~\cite{Shen2020Tcas}, SC flip (SCF)~\cite{Eran2020}, and SCL decoders~\cite{ren2022sequence}, respectively.}
\pf{Our area efficiency is $57.6\%$ lower than~\cite{Abbas17High}, but we see a significant $1.25$ dB SNR improvement when $\mathbb{L}=128$.}
Compared to other advanced BP decoders, such as BPF, the area efficiency of our work is \pf{$2.36\times$ and $2.74\times$} higher than that of~\cite{Shen2020Tcas} and~\cite{shen2020improved}, respectively.
\pf{Fig.~\ref{fig:performace_comp} plots the SNR that is required for a FER of $10^{-4}$ against area efficiency for various SOA polar decoders listed in Table~\ref{tab:tab_comp}, where our BPL decoder with $\mathbb{L}\in\{8\;32\;128\}$ achieves excellent performance, as shown by its placement in the top left corner of the graph.}
It is notable that, different from the SCL works~\cite{FastSSCL2017,ren2022sequence} that instantiate $\mathbb{L}$ independent SC decoders, our BPL decoder reuses a single BP decoder based on a serial decoding schedule and generates flexible permutations on-the-fly during the real-time decoding.
Therefore, even if the list size is increased, the area and operating frequency of our BPL decoder~are~not~affected.

\section{Conclusion}\label{sec:conclusion}
In this paper, we present an efficient BPL decoder implementation, which supports flexible permutation generation.
In terms of the algorithmic contributions, we propose a sequential generation algorithm to obtain a near-optimal PFG set.
Subsequently, we propose a hardware-friendly algorithm to generate flexible routings for permutations in hardware on-the-fly by a matrix decomposition.
On the architecture level, we present the BPL decoder with several optimizations to significantly reduce the hardware complexity and decoding latency, such as the flexible permutation generator and decoupled decoding schedule.
Synthesis results show that our BPL decoder can achieve a throughput of $25.63$ Gbps and an area efficiency of $29.46$ Gbps/mm$^2$ at SNR~$=4.0$~dB, which outperforms other existing SOA BP and BPF decoders.
Moreover, our BPL decoder efficiently implements the decoding on multiple PFGs and provides inspiration for \pf{the optimizations of long codes} and the implementation of generalized AE~decoding.

\bibliography{IEEEabrv,mybib}

\begin{thebibliography}{10}
\providecommand{\url}[1]{#1}
\csname url@samestyle\endcsname
\providecommand{\newblock}{\relax}
\providecommand{\bibinfo}[2]{#2}
\providecommand{\BIBentrySTDinterwordspacing}{\spaceskip=0pt\relax}
\providecommand{\BIBentryALTinterwordstretchfactor}{4}
\providecommand{\BIBentryALTinterwordspacing}{\spaceskip=\fontdimen2\font plus
\BIBentryALTinterwordstretchfactor\fontdimen3\font minus
  \fontdimen4\font\relax}
\providecommand{\BIBforeignlanguage}[2]{{%
\expandafter\ifx\csname l@#1\endcsname\relax
\typeout{** WARNING: IEEEtran.bst: No hyphenation pattern has been}%
\typeout{** loaded for the language `#1'. Using the pattern for}%
\typeout{** the default language instead.}%
\else
\language=\csname l@#1\endcsname
\fi
#2}}
\providecommand{\BIBdecl}{\relax}
\BIBdecl

\bibitem{Arikan2008Channel}
E.~Ar{\i}kan, ``Channel polarization: A method for constructing
  capacity-achieving codes,'' \emph{{IEEE} Trans. Inf. Theory}, vol.~55, no.~7,
  pp. 3051--3073, Jul. 2009.

\bibitem{5Gembb}
\emph{Chairman's Notes of Agenda Item 7.1.5 Channel coding and modulation},
  {3GPP TSG RAN WG1 Meeting \#87} {R}1-1613710, Nov. 2016.

\bibitem{tal2015list}
I.~Tal and A.~Vardy, ``List decoding of polar codes,'' \emph{{IEEE} Trans. Inf.
  Theory}, vol.~61, no.~5, pp. 2213--2226, Mar. 2015.

\bibitem{niu2012crc}
K.~Niu and K.~Chen, ``{CRC}-aided decoding of polar codes,'' \emph{{IEEE}
  Commun. Lett.}, vol.~16, no.~10, pp. 1668--1671, Sep. 2012.

\bibitem{hashemi2016TCAS1}
S.~A. Hashemi, C.~Condo, and W.~J. Gross, ``A fast polar code list decoder
  architecture based on sphere decoding,'' \emph{{IEEE} Trans. Circuits Syst.
  {I}}, vol.~63, no.~12, pp. 2368--2380, Dec. 2016.

\bibitem{FastSSCL2017}
S.~A. {Hashemi}, C.~{Condo}, and W.~J. {Gross}, ``Fast and flexible
  successive-cancellation list decoders for polar codes,'' \emph{{IEEE} Trans.
  Signal Process.}, vol.~65, no.~21, pp. 5756--5769, Aug. 2017.

\bibitem{hanif2018fast}
M.~Hanif, M.~H. Ardakani, and M.~Ardakani, ``Fast list decoding of polar codes:
  Decoders for additional nodes,'' in \emph{Proc. IEEE Wireless Commun. Netw.
  Conf. (WCNC)}, 2018, pp. 37--42.

\bibitem{shen2022}
Y.~Shen, Y.~Ren, A.~T. Kristensen, A.~Balatsoukas-Stimming, X.~You, C.~Zhang,
  and A.~P. Burg, ``Fast sequence repetition node-based successive cancellation
  list decoding for polar codes,'' in \emph{Proc. IEEE Int. Conf. Commun.
  (ICC)}, 2022, pp. 1--7.

\bibitem{ren2022sequence}
Y.~Ren, A.~T. Kristensen, Y.~Shen, A.~Balatsoukas-Stimming, C.~Zhang, and
  A.~Burg, ``A sequence repetition node-based successive cancellation list
  decoder for {5G} polar codes: Algorithm and implementation,'' \emph{{IEEE}
  Trans. Signal Process.}, vol.~70, pp. 5592--5607, 2022.

\bibitem{Ahmed2018}
A.~Elkelesh, M.~Ebada, S.~Cammerer, and S.~Ten~Brink, ``Belief propagation
  decoding of polar codes on permuted factor graphs,'' in \emph{Proc. IEEE
  Wireless Commun. Netw. Conf. (WCNC)}, 2018, pp. 1--6.

\bibitem{elkelesh2018belief}
------, ``Belief propagation list decoding of polar codes,'' \emph{{IEEE}
  Commun. Lett.}, vol.~22, no.~8, pp. 1536--1539, Jun. 2018.

\bibitem{A2019Com}
A.~{C} and O.~{Gazi}, ``Noise-aided belief propagation list decoding of polar
  codes,'' \emph{{IEEE} Commun. Lett.}, vol.~23, no.~8, pp. 1285--1288, May
  2019.

\bibitem{Ren2019ASICON}
Y.~Ren, W.~Xu, Z.~Zhang, X.~You, and C.~Zhang, ``Efficient belief propagation
  list decoding of polar codes,'' in \emph{Proc. IEEE Int. Conf. ASIC
  (ASICON)}, 2019, pp. 1--4.

\bibitem{Ahmed2020CRCBPL}
M.~Geiselhart, A.~Elkelesh, M.~Ebada, S.~Cammerer, and S.~Ten~Brink,
  ``{CRC}-aided belief propagation list decoding of polar codes,'' in
  \emph{Proc. IEEE Int. Symp. Inf. Theory (ISIT)}, 2020, pp. 395--400.

\bibitem{Ren2020}
Y.~{Ren}, Y.~{Shen}, Z.~{Zhang}, X.~{You}, and C.~{Zhang}, ``Efficient belief
  propagation polar decoder with loop simplification based factor graphs,''
  \emph{{IEEE} Trans. Veh. Technol.}, vol.~69, no.~5, pp. 5657--5660, Mar.
  2020.

\bibitem{Bai2020ISIT}
B.~Li, B.~Bai, M.~Zhu, and S.~Zhou, ``Improved belief propagation list decoding
  for polar codes,'' in \emph{Proc. IEEE Int. Symp. Inf. Theory (ISIT)}, 2020,
  pp. 1--6.

\bibitem{Nghia2018}
N.~Doan, S.~A. Hashemi, M.~Mondelli, and W.~J. Gross, ``On the decoding of
  polar codes on permuted factor graphs,'' in \emph{Proc. IEEE Global Commun.
  Conf. (GLOBECOM)}, 2018, pp. 1--6.

\bibitem{feng2021novel}
B.~Feng, R.~Liu, and K.~Tian, ``A novel post-processing method for belief
  propagation list decoding of polar codes,'' \emph{{IEEE} Commun. Lett.},
  vol.~25, no.~8, pp. 2468--2471, Jun. 2021.

\bibitem{yu2019belief}
Y.~Yu, Z.~Pan, N.~Liu, and X.~You, ``Belief propagation bit-flip decoder for
  polar codes,'' \emph{IEEE Access}, vol.~7, pp. 10\,937--10\,946, Jan. 2019.

\bibitem{Shen2020Tcas}
Y.~{Shen}, W.~{Song}, Y.~{Ren}, H.~{Ji}, X.~{You}, and C.~{Zhang}, ``Enhanced
  belief propagation decoder for {5G} polar codes with bit-flipping,''
  \emph{{IEEE} Trans. Circuits Syst. {II}}, vol.~67, no.~5, pp. 901--905, Mar.
  2020.

\bibitem{shen2020improved}
Y.~Shen, W.~Song, H.~Ji, Y.~Ren, C.~Ji, X.~You, and C.~Zhang, ``Improved belief
  propagation polar decoders with bit-flipping algorithms,'' \emph{{IEEE}
  Trans. Commun.}, vol.~68, no.~11, pp. 6699--6713, Aug. 2020.

\bibitem{Ji2020TCAS1}
H.~{Ji}, Y.~{Shen}, W.~{Song}, Z.~{Zhang}, X.~{You}, and C.~{Zhang}, ``Hardware
  implementation for belief propagation flip decoding of polar codes,''
  \emph{{IEEE} Trans. Circuits Syst. {I}}, vol.~68, no.~3, pp. 1330--1341, Dec.
  2020.

\bibitem{Stimming2017WCNC}
A.~Balatsoukas-Stimming, P.~Giard, and A.~Burg, ``Comparison of polar decoders
  with existing low-density parity-check and turbo decoders,'' in \emph{Proc.
  IEEE Wireless Commun. Netw. Conf. (WCNC)}, 2017, pp. 1--6.

\bibitem{alamdar2011simplified}
A.~Alamdar-Yazdi and F.~R. Kschischang, ``A simplified successive-cancellation
  decoder for polar codes,'' \emph{{IEEE} Commun. Lett.}, vol.~15, no.~12, pp.
  1378--1380, Dec. 2011.

\bibitem{sarkis2014fast}
G.~Sarkis, P.~Giard, A.~Vardy, C.~Thibeault, and W.~J. Gross, ``Fast polar
  decoders: Algorithm and implementation,'' \emph{{IEEE} J. Sel. Areas
  Commun.}, vol.~32, no.~5, pp. 946--957, May 2014.

\bibitem{hanif2017fast}
M.~Hanif and M.~Ardakani, ``Fast successive-cancellation decoding of polar
  codes: Identification and decoding of new nodes,'' \emph{{IEEE} Commun.
  Lett.}, vol.~21, no.~11, pp. 2360--2363, Nov. 2017.

\bibitem{condo2018generalized}
C.~Condo, V.~Bioglio, and I.~Land, ``Generalized fast decoding of polar
  codes,'' in \emph{Proc. IEEE Global Comm. Conf. (Globecom)}, 2018, pp. 1--6.

\bibitem{zheng2021threshold}
H.~Zheng, S.~A. Hashemi, A.~Balatsoukas-Stimming, Z.~Cao, T.~Koonen, J.~Cioffi,
  and A.~Goldsmith, ``Threshold-based fast successive-cancellation decoding of
  polar codes,'' \emph{{IEEE} Trans. Commun.}, vol.~69, no.~6, pp. 3541--3555,
  Jun. 2021.

\bibitem{you2021towards}
X.~You, C.-X. Wang, J.~Huang \emph{et~al.}, ``Towards {6G} wireless
  communication networks: Vision, enabling technologies, and new paradigm
  shifts,'' \emph{Sci. China Inf. Sci.}, vol.~64, no.~1, pp. 1--74, 2021.

\bibitem{Youn14belief}
Y.~S. Park, Y.~Tao, S.~Sun, and Z.~Zhang, ``A 4.68{Gb/s} belief propagation
  polar decoder with bit-splitting register file,'' in \emph{Proc. IEEE Symp.
  VLSI Circuits Digest Techni. Papers (VLSIC)}, 2014.

\bibitem{Jing2017WCSP}
S.~Jing, J.~Yang, A.~Yu, X.~You, and C.~Zhang, ``Graph-merged detection and
  decoding of polar-coded {MIMO} systems,'' in \emph{Proc. Int. Conf. Wireless
  Commun. Signal Process. (WCSP)}, 2017.

\bibitem{Amin2020TWC}
A.~Jalali and Z.~Ding, ``Joint detection and decoding of polar coded {5G}
  control channels,'' \emph{{IEEE} Trans. Wireless Commun.}, vol.~19, no.~3,
  pp. 2066--2078, Jan. 2020.

\bibitem{Hussami2009}
N.~{Hussami}, S.~B. {Korada}, and R.~{Urbanke}, ``Performance of polar codes
  for channel and source coding,'' in \emph{Proc. IEEE Int. Symp. Inf. Theory
  (ISIT)}, 2009, pp. 1488--1492.

\bibitem{Pamuk2012An}
A.~Pamuk, ``{An FPGA implementation architecture for decoding of polar
  codes},'' in \emph{Proc. IEEE Int. Symp. Wireless Commun. Syst. (ISWCS)},
  2012, pp. 437--441.

\bibitem{Bo2014Early}
B.~Yuan and K.~K. Parhi, ``Early stopping criteria for energy-efficient
  low-latency belief-propagation polar code decoders,'' \emph{{IEEE} Trans.
  Signal Process.}, vol.~62, no.~24, pp. 6496--6506, Dec. 2014.

\bibitem{Yang2016ISCAS}
J.~Yang, C.~Zhang, H.~Zhou, and X.~You, ``Pipelined belief propagation polar
  decoders,'' in \emph{Proc. IEEE Int. Symp. Circuits Syst. (ISCAS)}, 2016, pp.
  413--416.

\bibitem{Sun2016ISCAS}
S.~Sun and Z.~Zhang, ``Architecture and optimization of high-throughput belief
  propagation decoding of polar codes,'' in \emph{Proc. IEEE Int. Symp.
  Circuits Syst. (ISCAS)}, 2016, pp. 1--4.

\bibitem{Abbas17High}
S.~M. Abbas, Y.~Fan, J.~Chen, and C.-Y. Tsui, ``High-throughput and
  energy-efficient belief propagation polar code decoder,'' \emph{{IEEE} Trans.
  {VLSI} Syst.}, vol.~25, no.~3, pp. 1098--1111, Nov. 2017.

\bibitem{chen2019}
Y.~{Chen}, W.~{Sun}, C.~{Cheng}, T.~{Tsai}, Y.~{Ueng}, and C.~{Yang}, ``An
  integrated message-passing detector and decoder for polar-coded massive
  {MU-MIMO} systems,'' \emph{{IEEE} Trans. Circuits Syst. {I}}, vol.~66, no.~3,
  pp. 1205--1218, Nov. 2019.

\bibitem{geiselhart2022polar}
M.~Geiselhart, A.~Elkelesh, J.~Clausius, and S.~Ten~Brink, ``A polar subcode
  approach to belief propagation list decoding,'' pp. 243--248, 2022.

\bibitem{benevs1964optimal}
V.~E. Bene{\v{s}}, ``Optimal rearrangeable multistage connecting networks,''
  \emph{Bell System Technical Journal}, vol.~43, no.~4, pp. 1641--1656, Jul.
  1964.

\bibitem{Daesum2010TVLSI}
D.~Oh and K.~K. Parhi, ``Low-complexity switch network for reconfigurable
  {LDPC} decoders,'' \emph{{IEEE} Trans. {VLSI} Syst.}, vol.~18, no.~1, pp.
  85--94, Mar. 2010.

\bibitem{Geise2021TCOM}
M.~Geiselhart, A.~Elkelesh, M.~Ebada, S.~Cammerer, and S.~Ten~Brink,
  ``Automorphism ensemble decoding of {Reed–Muller} codes,'' \emph{{IEEE}
  Trans. Commun.}, vol.~69, no.~10, pp. 6424--6438, Jul. 2021.

\bibitem{bioglio2022group}
V.~Bioglio, I.~Land, and C.~Pillet, ``Group properties of polar codes for
  automorphism ensemble decoding,'' \emph{arXiv preprint arXiv:2206.03342},
  2022.

\bibitem{Su2022JSSCL}
B.-S. Su, C.-H. Lee, and T.-D. Chiueh, ``{A 58.6/91.3 pJ/b} dual mode belief
  propagation decoder for {LDPC} and polar codes in the {5G} communications
  standard,'' \emph{{IEEE} Solid State Circuits Lett.}, vol.~5, pp. 98--101,
  2022.

\bibitem{NR5G}
\emph{{5G NR: Multiplexing and channel coding}}, 3GPP Technical Specification
  38.212 Version 15.2.0, Jul. 2018.

\bibitem{Ahmed2019Tcom}
A.~Elkelesh, M.~Ebada, S.~Cammerer, and S.~Ten~Brink, ``Decoder-tailored polar
  code design using the genetic algorithm,'' \emph{{IEEE} Trans. Commun.},
  vol.~67, no.~7, pp. 4521--4534, Apr. 2019.

\bibitem{Kschichang2001FG}
F.~Kschischang, B.~Frey, and H.-A. Loeliger, ``Factor graphs and the
  sum-product algorithm,'' \emph{{IEEE} Trans. Inf. Theory}, vol.~47, no.~2,
  pp. 498--519, Feb. 2001.

\bibitem{arikan2010ISBC}
E.~Ar{\i}kan, ``Polar codes: {A} pipelined implementation,'' in \emph{Int.
  Symp. Broad. Commun. (ISBC)}, 2010, pp. 11--14.

\bibitem{xu2015xj}
J.~Xu, T.~Che, and G.~Choi, ``{XJ-BP}: Express journey belief propagation
  decoding for polar codes,'' in \emph{Proc. IEEE Global Comm. Conf.
  (Globecom)}, 2015, pp. 1--6.

\bibitem{dl2020Xu}
W.~Xu, X.~Tan, Y.~Be'ery, Y.-L. Ueng, Y.~Huang, X.~You, and C.~Zhang, ``Deep
  learning-aided belief propagation decoder for polar codes,'' \emph{{IEEE} J.
  Emer. Top. Circuits Syst.}, vol.~10, no.~2, pp. 189--203, May 2020.

\bibitem{Ren2015Asicon}
Y.~Ren, C.~Zhang, X.~Liu, and X.~You, ``Efficient early termination schemes for
  belief-propagation decoding of polar codes,'' in \emph{Proc. IEEE Int. Conf.
  on ASIC (ASICON)}, 2015, pp. 1--4.

\bibitem{Stimming_2015}
A.~Balatsoukas-Stimming, M.~{Bastani Parizi}, and A.~{Burg}, ``{LLR}-based
  successive cancellation list decoding of polar codes,'' \emph{{IEEE} Trans.
  Signal Process.}, vol.~63, no.~19, pp. 5165--5179, Oct. 2015.

\bibitem{Eran2020}
F.~{Ercan}, T.~{Tonnellier}, and W.~J. {Gross}, ``Energy-efficient hardware
  architectures for fast polar decoders,'' \emph{{IEEE} Trans. Circuits Syst.
  {I}}, vol.~67, no.~1, pp. 322--335, Oct. 2020.

\bibitem{Lee2020TSP}
H.-Y. Lee, Y.-H. Pan, and Y.-L. Ueng, ``A node-reliability based {CRC}-aided
  successive cancellation list polar decoder architecture combined with
  post-processing,'' \emph{{IEEE} Trans. Signal Process.}, vol.~68, pp.
  5954--5967, Oct. 2020.

\end{thebibliography}
\bibliographystyle{IEEEtran}
\end{document}